\documentclass[11pt,a4paper]{article}
\usepackage{fullpage}
\usepackage{amssymb,latexsym,amsmath,amsfonts,amsthm}
\usepackage{graphicx}
\usepackage{color}

\usepackage{overpic}
\usepackage{amsmath}
\usepackage{amsthm}
\usepackage{amssymb}
\usepackage{amsfonts}
\usepackage{hyperref}
\usepackage{indentfirst}
\usepackage{enumitem}
\usepackage{extarrows}

\renewcommand{\Re}{\mathrm{Re}\,}

\newcommand{\ud}{\,\mathrm{d}}
\newtheorem{thm}{Theorem}[section]

\newtheorem{prop}[thm]{Proposition}

\numberwithin{equation}{section}

\newcommand{\tr}{\mathrm{tr}}

\newcommand{\eq}{\begin{equation}}
\newcommand{\nq}{\end{equation}}
\newcommand{\eqa}{\begin{eqnarray}}
\newcommand{\nqa}{\end{eqnarray}}

\theoremstyle{remark}

\newtheorem{remark}[thm]{Remark}

\begin{document}

\title{On Wright's generalized Bessel kernel}
\author{Lun Zhang\footnotemark[1]}
\maketitle
\renewcommand{\thefootnote}{\fnsymbol{footnote}}
\footnotetext[1] {School of Mathematical Sciences and Shanghai Key Laboratory for Contemporary Applied Mathematics, Fudan University, Shanghai 200433, P. R. China. E-mail:
lunzhang\symbol{'100}fudan.edu.cn}

\begin{abstract}
In this paper, we consider the Wright's generalized Bessel kernel $K^{(\alpha,\theta)}(x,y)$ defined by
$$\theta x^{\alpha}\int_0^1J_{\frac{\alpha+1}{\theta},\frac{1}{\theta}}(ux)J_{\alpha+1,\theta}((uy)^{\theta})u^\alpha\ud u, \qquad \alpha>-1, \qquad \theta>0,$$
where $$J_{a,b}(x)=\sum_{j=0}^\infty\frac{(-x)^j}{j!\Gamma(a+bj)},\qquad a\in\mathbb{C},\qquad b>-1,$$ is Wright's generalization of the Bessel function. This non-symmetric kernel, which generalizes the classical Bessel kernel (corresponding to $\theta=1$) in random matrix theory, is the hard edge scaling limit of the correlation kernel for certain Muttalib-Borodin ensembles. We show that, if $\theta$ is rational, i.e., $\theta=\frac{m}{n}$ with $m,n\in\mathbb{N}$, $gcd(m,n)=1$, and $\alpha > m-1-\frac{m}{n}$, the Wright's generalized Bessel kernel is integrable in the sense of Its-Izergin-Korepin-Slavnov. We then come to the Fredholm determinant of this kernel over the union of several scaled intervals, which can also be interpreted as the gap probability (the probability of finding no particles) on these intervals. The integrable structure allows us to obtain a system of coupled partial differential equations associated with the corresponding Fredholm determinant as well as a Hamiltonian interpretation. As a consequence, we are able to represent the gap probability over a single interval $(0,s)$ in terms of a solution of a system of nonlinear ordinary differential equations.
\end{abstract}


\section{Introduction and statement of results} \label{sec:Intro}

\subsection{Muttalib-Borodin ensembles and the Wright's generalized Bessel kernel}
In 1995, Muttalib proposed the following class of joint probability density functions
\begin{equation}\label{eq:jpdf}
\frac{1}{Z_n}\Delta(x_1,\ldots,x_n)\Delta(x_1^\theta,\ldots,x_n^\theta)\prod_{j=1}^n w(x_j), \qquad \theta>0,\qquad x_j>0,
\end{equation}
for $n$ particles $x_1,\ldots , x_n$, where $w$ is a weight function over the positive real axis,
$$
Z_n=\int_{[0,\infty)^n}\Delta(x_1,\ldots,x_n)\Delta(x_1^\theta,\ldots,x_n^\theta)\prod_{j=1}^nw(x_j)\ud x_j
$$
is the normalization constant, and
$$
\Delta(\lambda_1,\ldots,\lambda_n)=\prod_{1\leq i <j \leq n}(\lambda_j-\lambda_i)
$$
is the standard Vandermonde determinant. Due to the appearance of two body interaction term $\Delta(x_1^\theta,\ldots,x_n^\theta)$, it was pointed out in \cite{Mut95} that these ensembles provide more effective description of disordered conductors in the metallic regime than the classical random matrix theory. A more concrete physical example that leads to \eqref{eq:jpdf} (with $\theta=2$, $w(x)=x^\alpha e^{-x}$) can be found in \cite{LSZ}, where the authors proposed a random matrix model for disordered bosons. These ensembles are further studied by Borodin \cite{Borodin99} under a more general framework, namely, biorthogonal ensembles. Following the conventions used in \cite{Forrester-Wang15}, we will refer to the class of joint probability density functions \eqref{eq:jpdf} as Muttalib-Borodin ensembles.

Muttalib-Borodin ensembles have drawn much attention recently. In \cite{Cheliotis14}, Cheliotis constructed certain triangular random matrices in terms of a Wishart matrix whose squared singular values are distributed according to \eqref{eq:jpdf} for the Laguerre weight; see also \cite{Forrester-Wang15} for the Jacobi case. Note that when $\theta=1$ and $w(x)=x^\alpha e^{-x}$, \eqref{eq:jpdf} reduces to the well-known Wishart-Laguerre unitary ensemble and plays a fundamental role in random matrix theory; cf. \cite{Anderson-Guionnet-Zeitouni10,Forrester10}. The limiting mean distribution of Muttalib-Borodin ensembles can be described by a minimizer of an equilibrium problem; see \cite{BLTW,Butez,CR14,Kuijlaars16} for the general weights and \cite{ESS,Forrester-Liu14,FLZ} for the special weights. The local behavior for the Laguerre weight can be found in \cite{Zhang15}.

From \eqref{eq:jpdf}, it is readily seen that they form the so-called determinantal point processes. This means there exits a correlation kernel $K_n^{(\alpha, \theta)}(x,y)$ such that the density functions can be rewritten as the following determinantal forms
\begin{equation*}
\frac{1}{n!}\det\left(K_n(x_i,x_j)\right)_{i,j=1}^n.
\end{equation*}
The kernel $K_n(x,y)$ has a representation in terms of biorthogonal polynomials (cf. \cite{Kon65} for a definition). Let
\begin{equation}
p_j(x)=\kappa_j x^j+\ldots, \qquad q_k(x)= x^k+\ldots, \quad \kappa_j > 0,
\end{equation}
be two sequences of polynomials of degree $j$ and $k$ respectively, and they satisfy the orthogonality conditions
\begin{equation}\label{eq:bioOP}
\int_0^\infty p_j(x)q_k(x^\theta)w(x)\ud x=\delta_{j,k}, \qquad j,k=0,1,2,\ldots.
\end{equation}
Note that the polynomial $q_k$ is normalized to be monic. We then have
\begin{equation}\label{eq:kerBio}
K_n(x,y)=\sum_{j=0}^{n-1}p_j(x)q_j(y^\theta )
w(x).
\end{equation}

For the cases that $w$ is a special Jacobi weight
$$w(x)=x^\alpha, \qquad x\in(0,1), \qquad \alpha>-1,$$
or a Laguerre weight
$$w(x)=x^\alpha e^{-x},\qquad \quad x>0, \qquad \alpha>-1,$$
the scaling limit of $K_n$ at the origin (hard edge scaling limit) is given by a new kernel $K^{(\alpha,\theta)}(x,y)$ depending on the parameters $\alpha$ and $\theta$ \cite{Borodin99}. For example, let $K_n^{\textrm{Lag}}$ be the correlation kernel corresponding to the Laguerre weight, it was shown by Borodin \cite[Theorem 4.2]{Borodin99} that
\begin{align}
\lim_{n \to \infty}\frac{K_n^{\textrm{Lag}}\left(\frac{x}{n^{1/\theta}},\frac{y}{n^{1/\theta }}\right)}{n^{1/\theta}}
= K^{(\alpha,\theta)}(x,y).
\end{align}
The kernel $K^{(\alpha,\theta)}(x,y)$ has the following explicit representations:
\begin{align}\label{eq:BorHard}
K^{(\alpha,\theta)}(x,y)&=\sum_{k,l=0}^{\infty}\frac{(-1)^kx^{\alpha+k}}{k!\Gamma\left(\frac{\alpha+1+k}{\theta}\right)}\frac{(-1)^ly^{\theta l}}{l!\Gamma(\alpha+1+\theta l)}\frac{\theta}{\alpha+1+k+\theta l}
\nonumber \\
&=\theta x^{\alpha}\int_0^1J_{\frac{\alpha+1}{\theta},\frac{1}{\theta}}(ux)J_{\alpha+1,\theta}((uy)^{\theta})u^\alpha\ud u
\nonumber \\
&=\frac{\theta }{(2\pi i)^2} \int_{c-i\infty}^{c+i\infty} \ud s \oint_{\Sigma}  \ud t
   \frac{\Gamma(s+1)\Gamma(\alpha+1+\theta s )}{\Gamma(t+1)\Gamma(\alpha+1+\theta t)}
            \frac{\sin\pi s}{\sin \pi t}       \frac{x^{-\theta s -1}y^{\theta t }}{s-t},
\end{align}
where  $J_{a,b}$ is Wright's generalization of the Bessel function \cite{Er53} given by
\begin{equation}\label{eq:Wright}
J_{a,b}(x)\footnote{In its original notation $\phi$ of \cite{Wright33}, one has $J_{a,b}(x)=\phi(b,a;-x)$.}=\sum_{j=0}^\infty\frac{(-x)^j}{j!\Gamma(a+bj)},\qquad a\in\mathbb{C},\qquad b>-1,
\end{equation}
$$c=\frac{\max\{0,1-\frac{\alpha+1}{\theta}\}-1}{2}<0,$$
and $\Sigma$ is a contour starting from $+\infty$ in the upper half plane and returning to $+\infty$ in the lower half plane which encircles the positive real axis with $\Re t>c$ for $t\in\Sigma$. In \eqref{eq:BorHard}, the first two formulas are given in \cite[Equation (3.6)]{Borodin99}\footnote{Due to a slightly different choice of the correlation kernel \eqref{eq:kerBio}, we have an extra factor $x^\alpha$ here.}, and the contour integral representation can be found in recent works \cite[Proposition 1.4.]{Forrester-Wang15} and \cite[Corollary 1.2]{Zhang15}.

The non-symmetric hard edge scaling limit generalizes the classical Bessel kernel \cite{Forrester93,TW94} (corresponding to $\theta=1$), and due to the appearance of Wright's generalized Bessel function, we call $K^{(\alpha,\theta)}(x,y)$ the Wright's generalized Bessel kernel in this paper. When $\theta=M\in\mathbb{N}=\{1,2,3,\ldots\}$ or $1/\theta = M$, the Wright's generalized Bessel kernel is related to the Meijer G-kernel arising from products of $M$ Ginibre matrices \cite{Kuijlaars-Zhang14}, as shown in \cite{Kuijlaars-Stivigny14}.

It is the aim of this paper to study the Wright's generalized Bessel kernel and the associated Fredholm determinant under the condition that $\theta\in\mathbb{Q}$, i.e., $\theta$ is a rational number. Our results are stated in the next few sections.


\subsection{Integrable representation of $K^{(\alpha,\theta)}$}\label{sec:integrable}
Recall that an integral operator with kernel $K(x,y)$ is called integrable in the sense of Its-Izergin-Korepin-Slavnov if
\[
K(x,y)=\frac{\sum_{i=1}^n f_i(x)g_i(y)}{x-y}, \qquad \text{with} \quad
\sum_{i=1}^nf_i(x)g_i(x)=0,
\]
for some $n \in \{2, 3, \ldots \}$, and certain  functions $f_i$ and $g_i$; see
\cite{IIKS90}. The kernels of standard universality classes (sine,
Airy, Bessel) encountered in random matrix theory all belong to the class of
integrable operators. It comes out that the Wright's generalized Bessel kernel $K^{(\alpha,\theta)}$ is integrable as well for special parameters as stated in the following theorem.

\begin{thm}\label{prop:integrable}
Let $\theta=\frac{m}{n}\in\mathbb{Q}$, where $m,n\in\mathbb{N}$ and $gcd(m,n)=1$. If $\alpha > m-1-\frac{m}{n}\geq -1$, we have, with the Wright's generalized Bessel kernel $K^{(\alpha,\theta)}$ defined in \eqref{eq:BorHard},
\begin{align}\label{eq:integrablerepr}
K^{(\alpha,\frac{m}{n})}(x,y)=m^mn^{n-1} x^{m-1} \frac{\mathcal {B}\left( x^{\alpha+1-m}J_{\frac{(\alpha+1)n}{m},\frac{n}{m}}(x), J_{\alpha+1,\frac{m}{n}}(y^{\frac{m}{n}}) \right)}{x^m-y^m},
    \end{align}
where $\mathcal {B}(\cdot,\cdot)$ is the bilinear
concomitant \cite{Ince} defined by
\begin{align}\label{def:bilinear}
\mathcal {B}\left(f(x),g(y)\right)=(-1)^{n+1}\sum_{j=0}^{m+n-1}(-1)^{j}
\left(\Delta_x\right)^j f(x)\sum_{i=0}^{m+n-1-j}\frac{b_{i+j}}{m^{i+j}}
\left(\Delta_y\right)^{i}g(y)
\end{align}
with $\Delta_x=x\frac{\ud}{\ud x}$ and $\Delta_y=y\frac{\ud}{\ud
y}$. The constants $b_i$ in \eqref{def:bilinear} are determined by
\begin{equation}\label{def:ai}
\prod_{i=1}^{m+n-1}(x-\nu_i)=\sum_{i=0}^{m+n-1} b_i x^i,
\end{equation}
with
\begin{equation}\label{def:nui}
\nu_i=\left\{
        \begin{array}{ll}
          \frac{i}{n}, & \hbox{$i=1,\ldots,n-1$,} \\
          1-\frac{\alpha}{m}-\frac{i-n+1}{m}, & \hbox{$i=n,n+1,\ldots,m+n-1$,}
        \end{array}
      \right.
\end{equation}
that is,
\begin{equation}\label{eq:elementary}
b_i=(-1)^ie_{m+n-1-i}(\nu_1,\ldots,\nu_{m+n-1})
\end{equation}
with $e_{i}(\nu_1,\ldots,\nu_{m+n-1})$ being the elementary symmetric polynomial.

Equivalently, one has
\begin{align}\label{eq:integrablereprscale}
&K^{(\alpha,\frac{m}{n})}(mn^{\frac{n}{m}}x^{\frac{1}{m}},mn^{\frac{n}{m}}y^{\frac{1}{m}}) \nonumber \\
&=\frac{x^{1-\frac{1}{m}}}{n^{\frac{n}{m}}}\frac{\widetilde{\mathcal {B}}\left( G^{m,0}_{0,m+n}\left({- \atop -\nu_{m+n-1}, \ldots, -\nu_1, \nu_0} \Big{|}x\right), G^{n,0}_{0,m+n}\left({- \atop \nu_0, \nu_1, \ldots, \nu_{m+n-1} } \Big{|}
y\right)\right)}{x-y}
\nonumber \\
&=m^\alpha n^{\frac{n\alpha}{m}-1} x^{1-\frac{1}{m}}
\frac{\widetilde{\mathcal {B}}\left( x^{\frac{\alpha+1}{m}-1}J_{\frac{(\alpha+1)n}{m},\frac{n}{m}}(mn^{\frac{n}{m}}x^{\frac{1}{m}}), J_{\alpha+1,\frac{m}{n}}(m^{\frac{m}{n}}ny^{\frac{1}{n}}) \right)}{x-y}
    \end{align}
where $\nu_0=0$, $\nu_i$, $i\geq 1$ is given in \eqref{def:nui}, $G^{n,0}_{0,m+n}$ is the Meijer G-function (see \eqref{def:Meijer} below for a definition) and
\begin{align}\label{def:bilineartilde}
\widetilde{\mathcal {B}}\left(f(x),g(y)\right)=(-1)^{n+1}\sum_{j=0}^{m+n-1}(-1)^{j}
\left(\Delta_x\right)^j f(x)\sum_{i=0}^{m+n-1-j} b_{i+j}
\left(\Delta_y\right)^{i}g(y)
\end{align}
with the same $b_i$ used in \eqref{def:bilinear}.
\end{thm}

If $\theta=M\in\mathbb{N}$, the integrable form can be obtained by combining the results in \cite{Kuijlaars-Zhang14} and \cite{Kuijlaars-Stivigny14}.

\subsection{Fredholm determinant}\label{sec:FredholmDet}
Let $J\subseteq [0,+\infty)$ be a bounded interval over the positive real axis, and set $\mathcal{K}_{\alpha,\theta}$ to be the integral operator with kernel $K^{(\alpha,\theta)}(x,y)\chi_J(y)$ acting on the function space $L^2((0,\infty))$, where $\chi_J$ is the characteristic function of the interval $J$. This operator sends the function $f$ to $$\int_J K^{(\alpha,\theta)}(x,y)f(y) \ud y.$$
To emphasize the dependence on the interval $J$, we will occasionally write
$$\mathcal{K}^{(\alpha,\theta)}=\mathcal{K}^{(\alpha,\theta)}\Big{|}_J,$$
and the same rule applies to other operators.

Due to the determinantal structure \eqref{eq:jpdf}, the associated Fredholm determinant $\det (I-\mathcal{K}_{\alpha,\theta})$ gives us the gap probability (the probability of finding no particles) on the interval $J$ for the limiting process of certain Muttalib-Borodin ensembles. This fact also implies that the Fredholm determinant is well-defined, and the operator $I-\mathcal{K}_{\alpha,\theta}$ is invertible. It is well-known that, for many integrable correlation kernels arising from random matrix theory, the associated Fredholm determinants are related to systems of integrable differential equations \cite{TW94c}--\cite{TW93}. As we shall see, it is also the case for the Fredholm determinant associated with Wright's generalized Bessel kernel in Theorem \ref{prop:integrable}. Hence, we will assume that $\theta=\frac{m}{n}$ and $\alpha > m-1-\frac{m}{n}$ in the rest of this paper.

To this end, let $0\leq a_1<a_2 <\ldots <a_{2\ell}$. Given a Wright's generalized Bessel kernel $K^{(\alpha,\frac{m}{n})}(x,y)$, we denote by $J$ the union of scaled intervals $(mn^{\frac{n}{m}}a_{2j-1}^{\frac{1}{m}},mn^{\frac{n}{m}}a_{2j}^{\frac{1}{m}})$, i.e.,
\begin{equation}
J=
\bigcup_{j=1}^{\ell}(mn^{\frac{n}{m}}a_{2j-1}^{\frac{1}{m}},mn^{\frac{n}{m}}a_{2j}^{\frac{1}{m}}).
\end{equation}
By the definition of Fredholm determinant, it is readily seen from \eqref{eq:integrablerepr}, \eqref{eq:integrablereprscale} and a change of variables that
\begin{align}\label{eq:relKtoKtilde}
&\det \left(I- \mathcal{K}_{\alpha,\frac{m}{n}} \right ) \nonumber  \\
&=
1+\sum_{k=1}^\infty \frac{(-1)^k}{k!}\int_J \cdots \int_J \det(K^{(\alpha,\frac{m}{n})}(x_i,x_j))_{i,j=1}^k \ud x_1\cdots \ud x_k \quad (x_i \to m n^{\frac{n}{m}} x_i^{\frac{1}{m}})
\nonumber \\
&=1+\sum_{k=1}^\infty \frac{(-1)^k}{k!} \int_{\widetilde{J}} \cdots \int_{\widetilde{J}} \det(\widetilde {K}^{(\alpha,\frac{m}{n})}(x_i,x_j))_{i,j=1}^k \ud x_1 \cdots \ud x_k
\nonumber \\
&= \det \left(I- \widetilde {\mathcal{K}}_{\alpha,\frac{m}{n}} \right ),
\end{align}
where
\begin{equation}
\widetilde J=
\bigcup_{j=1}^{\ell}(a_{2j-1},a_{2j}),
\end{equation}
and $\widetilde {\mathcal{K}}_{\alpha,\frac{m}{n}}$ is an integral operator acting on $L^2([0,\infty])$ with kernel
\begin{align}\label{def:tildeker}
&\widetilde{K}^{(\alpha,\frac{m}{n})}(x,y)\chi_{\widetilde J}(y)
\nonumber
\\
&= \frac{\widetilde{\mathcal {B}}\left( G^{m,0}_{0,m+n}\left({- \atop -\nu_{m+n-1}, \ldots, -\nu_1, \nu_0} \Big{|}x\right), G^{n,0}_{0,m+n}\left({- \atop \nu_0, \nu_1, \ldots, \nu_{m+n-1} } \Big{|} y\right)\right)}{x-y}\chi_{\widetilde J}(y)
\nonumber
\\
&=\frac{\sum_{i=0}^{m+n-1}\phi_i(x)\psi_i(y)}{x-y}\chi_{\widetilde J}(y),
\end{align}
and where for $i=0,1,\ldots,m+n-1$,
\begin{align}
\phi_i(x): &= (-1)^{n+1-i}(\Delta_x)^iG^{m,0}_{0,m+n}\left({- \atop -\nu_{m+n-1}, \ldots, -\nu_1, \nu_0} \Big{|}x\right),
\label{def:phi}
\\
\psi_i(y): &= \sum_{j=0}^{m+n-1-i}b_{i+j}(\Delta_y)^j G^{n,0}_{0,m+n}\left({- \atop \nu_0, \nu_1, \ldots, \nu_{m+n-1} } \Big{|} y\right).
\label{def:psi}
\end{align}

Since $\widetilde{K}^{(\alpha,\frac{m}{n})}(x,y) \neq \widetilde{K}^{(\alpha,\frac{m}{n})}(y,x)$ in general, we denote by
$\widetilde {\mathcal{K}}_{\alpha,\frac{m}{n}}'$ the integral operator with kernel $\widetilde {K}^{(\alpha,\frac{m}{n})}(y,x)\chi_{\widetilde J}(y)$ acting on the space $L^2((0,\infty))$.

\subsection{The system of partial differential equations}
For $j=0,1,\ldots,m+n-1$ and $k=1,\ldots, 2\ell $, we introduce the quantities
\begin{equation}\label{def:xy}
x_{j,k}:=(I-\widetilde {\mathcal{K}}_{\alpha,\frac{m}{n}})^{-1}\phi_j(a_k),\qquad y_{j,k}:=(I-\widetilde {\mathcal{K}}_{\alpha,\frac{m}{n}}')^{-1}\psi_j(a_k),
\end{equation}
and
\begin{align}
u_{j}&:=(-1)^n \int_{\widetilde {J}}\phi_0(x)(I-\widetilde {\mathcal{K}}'_{\alpha,\frac{m}{n}})^{-1}\psi_j(x)\ud x + b_{j-1},
\label{def:ujk}
\\
v_{j}&:=(-1)^n \int_{\widetilde {J}}\phi_j(x)(I-\widetilde {\mathcal{K}}'_{\alpha,\frac{m}{n}})^{-1}\psi_{m+n-1}(x)\ud x, \label{def:vjk}
\end{align}
where $b_j$ depending on the parameters $\alpha$, $m$ and $n$ is defined in \eqref{eq:elementary} for $j\geq 0$ and $b_{-1}:=0$. Note that all these quantities are actually functions of $a=(a_0,\ldots,a_{2\ell})$, and they satisfy the following system of partial differential equations.

\begin{prop}\label{prop:PDEforxy}
The functions $x_{j,k}$, $y_{j,k}$, $u_{j}$ and $v_{j}$ satisfy the following system of coupled partial differential equations£º
\begin{itemize}
  \item For $1\leq k \neq i \leq 2 \ell$ and $0\leq j \leq m+n-1$, we have
        \begin{align}
        \frac{\partial x_{j,k}}{\partial a_i} & =(-1)^i \frac{x_{j,i}}{a_k-a_i}\sum_{l=0}^{m+n-1}x_{l,k}y_{l,i}, \label{eq:xjkpi}\\
        \frac{\partial y_{j,k}}{\partial a_i} & =(-1)^i \frac{y_{j,i}}{a_i-a_k}\sum_{l=0}^{m+n-1}x_{l,i}y_{l,k}.
\label{eq:yjipi}
        \end{align}

  \item For $1 \leq k \leq 2 \ell$ and $0 \leq j \leq m+n-2$, we have
\begin{align}\label{eq:xjkpk1}
a_k\frac{\partial x_{j,k}}{\partial a_k}=-v_jx_{0,k}-x_{j+1,k}-\sum_{i=1, i\neq k}^{2\ell}(-1)^i\frac{a_i x_{j,i}}{a_k-a_i}\sum_{l=0}^{m+n-1}x_{l,k}y_{l,i},
\end{align}
while for $j=m+n-1$, we have
\begin{multline}\label{eq:xjkpk2}
a_k\frac{\partial x_{m+n-1,k}}{\partial a_k}=((-1)^{n+1}a_k-v_{m+n-1})x_{0,k}+\sum_{i=0}^{m+n-1}u_ix_{i,k}
\\-\sum_{i=1, i\neq k}^{2\ell}(-1)^i\frac{a_i x_{m+n-1,i}}{a_k-a_i}\sum_{l=0}^{m+n-1}x_{l,k}y_{l,i}.
\end{multline}

  \item For $1 \leq k \leq 2 \ell$ and $1 \leq j \leq m+n-1$, we have
\begin{align}\label{eq:yjkpk1}
a_k\frac{\partial y_{j,k}}{\partial a_k}=y_{j-1,k}-u_jy_{m+n-1,k}-\sum_{i=1, i\neq k}^{2\ell}(-1)^i\frac{a_i y_{j,i}}{a_i-a_k}\sum_{l=0}^{m+n-1}x_{l,i}y_{l,k},
\end{align}
while for $j=0$, we have
\begin{multline}\label{eq:yjkpk2}
a_k\frac{\partial y_{0,k}}{\partial a_k}=\sum_{i=0}^{m+n-1}v_i y_{i,k}+((-1)^{n}a_k-u_{0})y_{m+n-1,k}
\\-\sum_{i=1, i\neq k}^{2\ell}(-1)^i\frac{a_i y_{0,i}}{a_i-a_k}\sum_{l=0}^{m+n-1}x_{l,i}y_{l,k}.
\end{multline}

\item For $1 \leq k \leq 2 \ell$ and $0 \leq j \leq m+n-1$, we have
\begin{align}
\frac{\partial u_j}{\partial a_k}&=(-1)^{n+k}x_{0,k}y_{j,k}, \label{eq:ujpartial}
\\
\frac{\partial v_j}{\partial a_k}&=(-1)^{n+k}x_{j,k}y_{m+n-1,k}. \label{eq:vjpartial}
\end{align}
\end{itemize}
\end{prop}

The partial differential equations are similar to those derived for the sine kernel by Jimbo, Miwa, M\^{o}ri and Sato \cite{JMMS}; for the Airy and Bessel kernels by Tracy and Widom  \cite{TW94,TW94a}, and for the Meijer G-kernel by Strahov \cite{Stro14}, as is their derivation.

As in \cite{TW93,Stro14}, we also have a Hamiltonian interpretation for the dynamical equations stated in Proposition \ref{prop:PDEforxy}. To this end, let us set
\begin{equation}\label{def:Hk}
H_{k}:=a_k \frac{\partial}{\partial a_k} \log \det \left(I- \widetilde {\mathcal{K}}_{\alpha,\frac{m}{n}} \right ),\quad 1\leq k \leq 2\ell,
\end{equation}
and
\begin{equation}\label{eq:xytopq}
x_{j,2k}= \imath q_{j,2k}, \quad x_{j,2k+1}= q_{j,2k+1}, \quad y_{j,2k}= \imath p_{j,2k}, \quad y_{j,2k+1}=  p_{j,2k+1}.
\end{equation}
In the $2\ell (m+n-1)$ canonical coordinates $$(p_{j,k},q_{j,k},u_j,v_j), \quad 1 \leq k \leq 2 \ell, \quad 0 \leq j \leq m+n-2,$$ for any two functions $f$ and $g$ of these variables, we define the Poisson bracket by
\begin{multline}
\{f,g\}=\sum_{k=1}^{2\ell}\frac{1}{a_k}\sum_{j=0}^{m+n-1}\left(\frac{\partial f}{\partial q_{j,k}}\frac{\partial g}{\partial p_{j,k}}-\frac{\partial f}{\partial p_{j,k}}\frac{\partial g}{\partial q_{j,k}}\right)
\\
+(-1)^n\sum_{j=0}^{m+n-1}\left(\frac{\partial f}{\partial u_j}\frac{\partial g}{\partial v_j}-\frac{\partial f}{\partial v_j}\frac{\partial g}{\partial u_j}\right).
\end{multline}
It is then readily seen that the following symplectic Poisson brackets hold:
\begin{align}
\{q_{j,k},q_{i,l}\}&=\{p_{j,k},p_{i,l}\}=0,\qquad\{q_{j,k},p_{i,l}\}=\frac{1}{a_k}\delta_{j,i}\delta_{k,l},
\\
\{ u_i,u_j \} & =\{v_i,v_j\}=0,\qquad \{u_i,v_j\}=(-1)^n\delta_{i,j}.
\end{align}
\begin{prop}\label{prop:Hamiltonian}
With the Hamiltonians $H_{k}$ defined in \eqref{def:Hk}, we have
\begin{multline}\label{eq:Hkexplicit}
H_k=\left(-\sum_{j=0}^{m+n-1}v_jp_{j,k}+(-1)^{n+1}a_k p_{m+n-1,k}\right)q_{0,k}+\sum_{j=0}^{m+n-1}u_jq_{j,k}p_{m+n-1,k}
\\-\sum_{j=0}^{m+n-2}q_{j+1,k}p_{j,k}+\sum_{i=1,i\neq k}^{2\ell}\frac{a_i}{a_k-a_i}\sum_{\kappa,j=0}^{m+n-1}q_{\kappa,k}
p_{\kappa,i}q_{j,i}p_{j,k},
\end{multline}
and they are in involution so that
\begin{equation}\label{eq:involution}
\{H_i, H_j\}=0, \qquad 1\leq i,j \leq 2\ell.
\end{equation}
Furthermore, the equations \eqref{eq:xjkpi}--\eqref{eq:vjpartial} can be written as
\begin{equation}\label{eq:pqpartialder}
\frac{\partial q_{j,k}}{\partial a_i}=\{q_{j,k},H_i\},\qquad \frac{\partial p_{j,k}}{\partial a_i}=\{p_{j,k},H_i\},
\end{equation}
and
\begin{equation}\label{eq:uvpartialder}
\frac{\partial u_j}{\partial a_i}=\{u_j,H_i\}, \qquad \frac{\partial v_j}{\partial a_i}=\{v_j,H_i\}.
\end{equation}
\end{prop}

The Hamiltonian interpretation particularly implies the complete integrability of the partial differential equations, as explained in \cite{TW93} for the sine kernel.

\subsection{Gap probability on a single interval $(0,s)$}
Our final result is concerned with the Fredholom determinant over a single interval, i.e., the probability of finding no particles on this interval.
\begin{thm}\label{prop:oneinterval}
For $s>0$, let
\begin{equation}\label{def:Fs}
F_{\alpha,\frac{m}{n}}(s):=\det \left(I- \mathcal{K}_{\alpha,\frac{m}{n}}\Big{|}_{(0,s)} \right ).
\end{equation}
Then, we have
\begin{align}\label{eq:Gapformular}
F_{\alpha,\frac{m}{n}}(s)
&=\exp\left(\frac{(-1)^n }{m^{m-2}n^n}\int_0^s\log\left(\frac{s}{\tau}\right)\tau^{m-1}x_0\left(\frac{\tau^m}{m^mn^n}\right)y_{m+n-1}\left(\frac{\tau^m}{m^mn^n}\right)\ud \tau\right)
\nonumber \\
&=\exp\left(m\int_0^s\frac{v_0\left(\frac{\tau^m}{m^mn^n}\right)}{\tau}\ud \tau \right),
\end{align}
where the functions $x_0$, $y_{m+n-1}$ and $v_0$ are solutions of the following system of $4(m+n)$ nonlinear ordinary differential equations:
\begin{align}
s\frac{\ud x_{j}(s)}{\ud s}&=-v_j(s)x_{0}(s)-x_{j+1}(s),\qquad 0 \leq j \leq m+n-2, \label{eq:eqforx}\\
s\frac{\ud x_{m+n-1}(s)}{\ud s}&=((-1)^{n+1}s-v_{m+n-1}(s))x_{0}(s)+\sum_{i=0}^{m+n-1}u_i(s)x_{i}(s),
\\
s\frac{\ud y_{j}(s)}{\ud s}&=y_{j-1}(s)-u_j(s)y_{m+n-1}(s), \qquad 1 \leq j \leq m+n-1,
\\
s\frac{\ud y_{0}(s)}{\ud s}&=\sum_{i=0}^{m+n-1}v_i(s) y_{i}(s)+((-1)^{n}s-u_{0}(s))y_{m+n-1}(s),
\\
\frac{\ud u_j(s)}{\ud s}&=(-1)^{n}x_{0}(s)y_{j}(s), \qquad 0 \leq j \leq m+n-1,
\\
\frac{\ud v_j(s)}{\ud s}&=(-1)^{n}x_{j}(s)y_{m+n-1}(s), \qquad 0 \leq j \leq m+n-1, \label{eq:eqforv}
\end{align}
with initial conditions
\begin{align}
& x_{j}(s)\sim \phi_j(s), \quad y_{j}(s)\sim \psi_j(s), \quad s \to 0, \label{eq:xyssmall}
\\
& u_{j}(0)=b_{j-1},
\quad
v_{j}(0)=0, \quad  0 \leq j \leq m+n-1.\label{eq:ICuv}
\end{align}

\end{thm}
The initial conditions \eqref{eq:xyssmall}--\eqref{eq:ICuv} follow directly from their definitions, but it is not clear whether the solutions with this behaviour are unique or not. We also provide some plots of the determinants in Figure \ref{fig:Fredholm}, using the strategy in \cite{Born10}.

\begin{figure}[t]
  \centering
  \begin{overpic}[scale=.8]{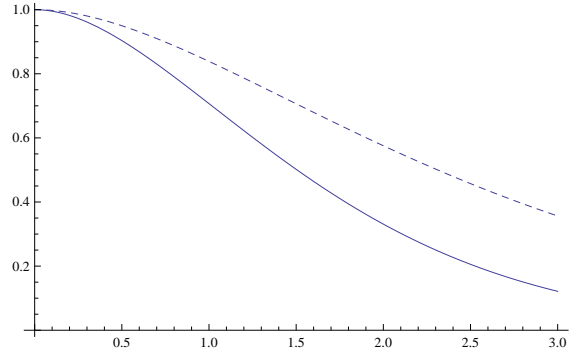}
  \end{overpic}
  \caption{Numerical computation of $F_{\alpha,\theta}(s)$ for $\theta=1,\alpha=1$ (dashed) and $\theta=2,\alpha=1$ (solid).}
  \label{fig:Fredholm}
\end{figure}

Since $\mathcal{K}_{\alpha,\frac{m}{n}}$ is equivalent to the Bessel kernel if $\theta=1$,  the system of ordinary differential equations should be related to a particular Painlev\'{e} III system as shown in \cite{TW94}; see also \cite{NF16} for an explanation. In general, it would be very interesting to see whether one could reduce a single differential equation satisfied by $v_0$ from \eqref{eq:eqforx}--\eqref{eq:eqforv}. This will in turn provide a higher order generalization of the third Painlev\'{e} equation, and be helpful to establish the large $s$ asymptotics of $F_{\alpha,\frac{m}{n}}$. The study of this aspect has been conducted in \cite{NF16} for the case that $\theta=2$, based on the results in \cite{Stro14}, where an explicit 4-th order differential equation has been derived and its relation to the so-called 4-dimensional Painlev\'{e}-type equation is discussed.

\begin{remark}
It was pointed out by one of the referees that the small $s$ asymptotics of $F_{\alpha,\theta}(s)$ can be derived quickly from the double contour integral representation of Wright's generalized Bessel kernel given in \eqref{eq:BorHard} for general $\theta$. Indeed, by \eqref{eq:relKtoKtilde} and \eqref{eq:BorHard}, it follows that, as $s\to 0$,
\begin{align}\label{eq:smalls1}
\log F_{\alpha,\theta}(s) &\sim - \int_0^s K^{(\alpha,\theta)}(x,x)\ud x \nonumber
\\
&=\frac{-\theta }{(2\pi i)^2} \int_{c-i\infty}^{c+i\infty} \ud \varrho \oint_{\Sigma}  \ud t
   \frac{\Gamma(\varrho+1)\Gamma(\alpha+1+\theta \varrho )}{\Gamma(t+1)\Gamma(\alpha+1+\theta t)}
            \frac{\sin\pi \varrho}{\sin \pi t}      \frac{ \int_{0}^s x^{\theta(t- \varrho) - 1}\ud x}{\varrho-t}
\nonumber
\\
&=\frac{1}{(2\pi i)^2} \int_{c-i\infty}^{c+i\infty} \ud \varrho \oint_{\Sigma}  \ud t
   \frac{\Gamma(\varrho+1)\Gamma(\alpha+1+\theta \varrho )}{\Gamma(t+1)\Gamma(\alpha+1+\theta t)}
            \frac{\sin\pi \varrho}{\sin \pi t}      \frac{s^{\theta(t-\varrho)}}{(\varrho-t)^2}.
\end{align}
To this end, we note that the integrand in \eqref{eq:smalls1} has poles in $t$ (due to $\sin(\pi t)$) at $0,1,2,\ldots$, and in $\varrho$ (due to $\Gamma(\alpha+1+\theta \varrho)$) at $-\frac{\alpha+1}{\theta}-\theta k$, $k=0,1,2,\ldots$. Since $s$ is small, the main contribution to the integral comes from the residues in $t$ and $\varrho$ that are closest to each other, which corresponds to $s=-\frac{\alpha+1}{\theta}$ and $t=0$. A straightforward calculation with the help of residue theorem then gives us that
\begin{align}\label{eq:smalls}
\log F_{\alpha,\theta}(s) &\sim
-\frac{\theta }{\pi} \frac{\Gamma\left(1-\frac{\alpha+1}{\theta}\right)}{\Gamma(\alpha+1)}
\sin\left(\frac{\pi(\alpha+1)}{\theta}\right)\frac{s^{\alpha+1}}{(\alpha+1)^2}
\nonumber
\\
&=-\frac{\theta }{\Gamma\left(\frac{\alpha+1}{\theta}\right)\Gamma(\alpha+1)} \frac{s^{\alpha+1}}{(\alpha+1)^2},\qquad s\to 0,
\end{align}
where we have made use of the Euler's reflection formula for the gamma function in the last step.
\end{remark}

\subsection{Organization of the rest of the paper}
The rest of this paper is organized as follows. The proof of Theorem \ref{prop:integrable} is given in Section \ref{sec:props}, following the strategy in \cite{Bertola-Gekhtman-Szmigielski14,Kuijlaars-Zhang14}. The essential issue is that, if $b$ is rational, the Wright's generalized Bessel function $J_{a,b}$ satisfies an ordinary differential equation and  is related to the Meijer G-function. The proofs of Propositions \ref{prop:PDEforxy}--\ref{prop:Hamiltonian} and Theorem \ref{prop:oneinterval} are presented in Section \ref{Sec:proofofPDE}, where we start with a review of some basic facts about the operators. Since the Wright's generalized Bessel kernel is not symmetric in general, we have to introduce extra quantities $y_{j,k}$ in \eqref{def:xy}, and the argument follows from Strahov's extension of Tracy-Widom theory \cite{Stro14}, dealing with the Meijer G-kernel. Although these two kernels are related to each other for special cases, we emphasize that they are not equivalent in general. For convenience of the readers, we include some basic properties of the Meijer G-function and the Wright's generalized Bessel function used throughout this paper in the Appendix.

\section{Proof of Theorem \ref{prop:integrable}}\label{sec:props}
We set
\begin{align}\label{def:pandq}
p^{(\alpha,\theta)}(x)=x^{\alpha}J_{\frac{\alpha+1}{\theta},\frac{1}{\theta}}(x),
\qquad
q^{(\alpha,\theta)}(y)=\theta J_{\alpha+1,\theta}(y^{\theta}).
\end{align}
By \eqref{eq:BorHard}, our task is then to evaluate the integral
\begin{equation} \label{eq:Knuintegral}
K^{(\alpha,\theta)}(x,y) = \int_0^1 p^{(\alpha,\theta)}(ux)q^{(\alpha,\theta)}(uy)\ud u,
\end{equation}
under the condition that $\theta=m/n\in\mathbb{Q}$, where $m,n\in\mathbb{N}$ and $(m,n)=1$.

In view of \eqref{eq:WrightinMeijer}, it is readily seen that
\begin{equation}\label{eq:qinMeijer}
q^{(\alpha,\frac{m}{n})}(y)=(2\pi)^{\frac{m-n}{2}}m^{-\alpha+\frac{1}{2}}n^{-\frac{1}{2}}G^{n,0}_{0,m+n}\left({- \atop \nu_0, \nu_1, \ldots, \nu_{m+n-1} }\Big{|}
\frac{y^m}{m^mn^n}\right),
\end{equation}
where $\nu_0=0$ and $\nu_i$, $i\geq 1$ is given in \eqref{def:nui}.
Similarly, for $p^{(\alpha,\frac{m}{n})}(x)$, it is readily seen from $\eqref{eq:WrightinMeijer}$ that
\begin{align*}
&p^{(\alpha,\frac{m}{n})}(m n^{\frac{n}{m}}x^{\frac{1}{m}})
\\
&=(2\pi)^{\frac{n-m}{2}}m^{\frac{1}{2}+\alpha}n^{\frac{1}{2}-\frac{n}{m}}x^{\frac{\alpha}{m}}
G^{m,0}_{0,m+n}\left({- \atop \frac{m-1}{m}, \ldots, 0, 1-\frac{\alpha+1}{m}-\frac{n-1}{n},\ldots,1-\frac{\alpha+1}{m}}\Big{|}
x\right)
\\
&=(2\pi)^{\frac{n-m}{2}}m^{\frac{1}{2}+\alpha}n^{\frac{1}{2}-\frac{n}{m}}
G^{m,0}_{0,m+n}\left({- \atop 1-\frac{1}{m}-\nu_{m+n-1}, \ldots, 1-\frac{1}{m}-\nu_0} \Big{|}
x\right),
\end{align*}
where in the last step we have made use of \eqref{eq:multiply}. Hence,
\begin{align}\label{eq:pinMeijer}
p^{(\alpha,\frac{m}{n})}(x)
=(2\pi)^{\frac{n-m}{2}}m^{\frac{1}{2}+\alpha}n^{\frac{1}{2}-\frac{n}{m}}
G^{m,0}_{0,m+n}\left({- \atop 1-\frac{1}{m}-\nu_{m+n-1}, \ldots, 1-\frac{1}{m}-\nu_0} \Big{|}
\frac{x^m}{m^mn^n}\right).
\end{align}

Substituting \eqref{eq:pinMeijer} and \eqref{eq:qinMeijer} into \eqref{eq:Knuintegral}, we obtain from \eqref{eq:multiply} that
\begin{align}
& K^{(\alpha,\frac{m}{n})}(x,y) \nonumber
\\
&= m n^{-\frac{n}{m}}\int_0^1 G^{m,0}_{0,m+n}\left({- \atop 1-\frac{1}{m}-\nu_{m+n-1}, \ldots, 1-\frac{1}{m}-\nu_1, 1-\frac{1}{m}-\nu_0} \Big{|}
\frac{u^mx^m}{m^mn^n}\right)
\nonumber
\\& \qquad \qquad \qquad \qquad \qquad \times G^{n,0}_{0,m+n}\left({- \atop \nu_0, \nu_1, \ldots, \nu_{m+n-1} }\Big{|}
\frac{u^my^m}{m^mn^n}\right)\ud u
\nonumber
\\&= m n^{-\frac{n}{m}}\int_0^1 \left(\frac{u^mx^m}{m^mn^n}\right)^{1-\frac{1}{m}}G^{m,0}_{0,m+n}\left({- \atop -\nu_{m+n-1}, \ldots, -\nu_1, \nu_0} \Big{|}\frac{u^mx^m}{m^mn^n}\right)
\nonumber
\\& \qquad \qquad \qquad \qquad \qquad \times G^{n,0}_{0,m+n}\left({- \atop \nu_0, \nu_1, \ldots, \nu_{m+n-1} }\Big{|}
\frac{u^my^m}{m^mn^n}\right)\ud u \nonumber
\\&=\frac{m^{1-m}}{n^n} x^{m-1}\int_0^1 G^{m,0}_{0,m+n}\left({- \atop -\nu_{m+n-1}, \ldots, -\nu_1, \nu_0} \Big{|}\frac{t x^m}{m^mn^n}\right)
\nonumber
\\& \qquad \qquad \qquad \qquad \qquad \times G^{n,0}_{0,m+n}\left({- \atop \nu_0, \nu_1, \ldots, \nu_{m+n-1} }\Big{|}
\frac{t y^m}{m^mn^n}\right)\ud t \qquad (t=u^m)
\nonumber
\\&=x^{m-1}\int_0^1f(t^{\frac{1}{m}}x)g(t^{\frac{1}{m}}y)\ud t, \label{eq:kinMeij}
\end{align}
where
\begin{multline}\label{def:f}
f(x)= (2\pi)^{\frac{n-m}{2}}m^{\frac{1}{2}+\alpha-m}n^{\frac{1}{2}-n}
G^{m,0}_{0,m+n}\left({- \atop -\nu_{m+n-1}, \ldots, -\nu_1, \nu_0} \Big{|}
\frac{x^m}{m^mn^n}\right)
\\=\frac{x^{1-m}}{m}p^{(\alpha,\frac{m}{n})}(x)=\frac{x^{\alpha+1-m}}{m}J_{\frac{(\alpha+1)n}{m},\frac{n}{m}}(x),
\end{multline}
and
\begin{multline}\label{def:g}
g(y)=(2\pi)^{\frac{m-n}{2}}m^{-\alpha+\frac{1}{2}}n^{-\frac{1}{2}}G^{n,0}_{0,m+n}\left({- \atop \nu_0, \nu_1, \ldots, \nu_{m+n-1} }\Big{|}
\frac{y^m}{m^mn^n}\right)\\
=q^{(\alpha,\frac{m}{n})}(y)=\frac{m}{n}J_{\alpha+1,\frac{m}{n}}(y^{\frac{m}{n}}).
\end{multline}

Note that the Meijer-G function satisfies the differential equation \eqref{eq:diff}.
For $f$ and $g$ defined \eqref{def:f} and \eqref{def:g}, it is readily seen that, for every $t$,
\begin{align}
g(t^{\frac{1}{m}}y) \prod_{j=0}^{m+n-1}\left(\frac{\Delta_x}{m}+\nu_j\right)f(t^\frac{1}{m}x)&=(-1)^m \frac{tx^m}{m^mn^n} f(t^{\frac{1}{m}}x)g(t^{\frac{1}{m}}y), \label{eq:gDf}
\\
f(t^{\frac{1}{m}}x) \prod_{j=0}^{m+n-1}\left(\frac{\Delta_y}{m}-\nu_j \right)g(t^\frac{1}{m}y)&=(-1)^n \frac{ty^m}{m^mn^n} g(t^{\frac{1}{m}}y)f(t^{\frac{1}{m}}x),\label{eq:fDg}
\end{align}
where $\Delta_x=x\frac{\ud}{\ud x}$ and $\Delta_y=y\frac{\ud}{\ud y}$.
If $(-1)^n=(-1)^m=-1$, we subtract these two identities, otherwise
we add them together. Since the arguments in the other cases are similar,
we restrict to the case that $(-1)^n=(-1)^m=-1$.

Subtracting \eqref{eq:gDf} from \eqref{eq:fDg} we obtain
\begin{align}
&\frac{x^m-y^m}{m^mn^n}f(t^{\frac{1}{m}}x)g(t^{\frac{1}{m}}y) \nonumber
\\
&=\frac{1}{t}\left(f(t^{\frac{1}{m}}x)\prod_{j=0}^{m+n-1}\left(\frac{\Delta_y}{m}-\nu_j\right)g(t^{\frac{1}{m}}y)-
g(t^{\frac{1}{m}}y)\prod_{j=0}^{m+n-1}\left(\frac{\Delta_x}{m}+\nu_j\right)f(t^{\frac{1}{m}}x)\right)
\nonumber
\\
&=\frac{1}{t}\sum_{i=0}^{m+n-1}
\frac{b_i}{m^{1+i}}\left(f(t^{\frac{1}{m}}x)(\Delta_y)^{i+1}g(t^{\frac{1}{m}}y)+(-1)^ig(t^{\frac{1}{m}}y)(\Delta_x)^{i+1}
f(t^{\frac{1}{m}}x)\right),
\label{eq:difference}
\end{align}
where the constants $b_i$ are defined in \eqref{def:ai} and
\eqref{eq:elementary}. To this end, we observe that
\begin{multline*}
\frac{\partial}{\partial t}\left(\sum_{j=0}^{i}(-1)^{j}
\left(\Delta_x\right)^jf(t^{\frac{1}{m}}x)\left(\Delta_y\right)^{i-j}g(t^{\frac{1}{m}}y)\right) \\
    = \frac{1}{m t} \left(f(t^{\frac{1}{m}}x)(\Delta_y)^{i+1}g(t^{\frac{1}{m}}y)+(-1)^ig(t^{\frac{1}{m}}y)(\Delta_x)^{i+1}
f(t^{\frac{1}{m}}x)\right),
\end{multline*}
thus,
\begin{align}
&\frac{x^m-y^m}{m^mn^n}f(t^{\frac{1}{m}}x)g(t^{\frac{1}{m}}y) \nonumber \\
&=\sum_{i=0}^{m+n-1}
\frac{b_i}{m^{i}}\frac{\partial}{\partial t}\left(\sum_{j=0}^{i}(-1)^{j}
\left(\Delta_x\right)^jf(t^{\frac{1}{m}}x)\left(\Delta_y\right)^{i-j}g(t^{\frac{1}{m}}y)\right)
\nonumber
\\
&=\frac{\partial}{\partial t} \sum_{j=0}^{m+n-1}(-1)^{j}
\left(\Delta_x\right)^j f(t^{\frac{1}{m}}x)\sum_{i=0}^{m+n-1-j}\frac{b_{i+j}}{m^{i+j}}
\left(\Delta_y\right)^{i}g(t^{\frac{1}{m}}y)= \frac{\partial}{\partial t} \mathcal B(f(t^{\frac{1}{m}}x), g(t^{\frac{1}{m}}y)).
\end{align}
This, together with \eqref{eq:kinMeij}, implies that
\[ \left(\frac{x^m-y^m}{m^mn^n}\right) K^{(\alpha,\frac{m}{n})}(x,y) = x^{m-1}\left( \mathcal B(f(x), g(y)) - \lim_{t \to 0+} \mathcal B(f(t^{\frac{1}{m}}x), g(t^{\frac{1}{m}}y))\right), \]
which gives us \eqref{eq:integrablerepr} provided we can show that
\begin{equation}  \label{show:limit}
    \lim_{t \to 0+} \mathcal B(f(t^{\frac{1}{m}}x), g(t^{\frac{1}{m}}y)) = 0.
    \end{equation}

To show \eqref{show:limit}, we first observe from \eqref{def:f}, \eqref{def:g} and \eqref{eq:Wright} that, as $t\to 0+$,
\begin{equation}
\left(\Delta_x\right)^j f(t^{\frac{1}{m}}x)= \left\{
                                               \begin{array}{ll}
                                                 \mathcal{O}(t^{\frac{\alpha+1-m}{m}}), & \hbox{$\alpha\neq m-1$, $j=0,1,2,\ldots$,} \\
                                                 \mathcal{O}(1), & \hbox{$\alpha= m-1$, $j=0$,} \\
                                                  \mathcal{O}(t^{\frac{1}{m}}), & \hbox{$\alpha = m-1$, $j=1,2,\ldots,$}
                                               \end{array}
                                             \right.
\end{equation}
and
\begin{equation}
\left(\Delta_y\right)^j g(t^{\frac{1}{m}}y)=\left\{
                                              \begin{array}{ll}
                                                \mathcal{O}(1), & \hbox{$j=0$,} \\
                                                \mathcal{O}(t^{\frac{1}{n}}), & \hbox{$j=1,2,\ldots.$}
                                              \end{array}
                                            \right.
 \end{equation}
The condition $\alpha> m-1-\frac{m}{n}$ now ensures that
$$\frac{\alpha+1-m}{m}+\frac{1}{n}>0.$$
Hence, the terms in $\mathcal{B}(f(t^{\frac{1}{m}}x), g(t^{\frac{1}{m}}y))$ containing $\left(\Delta_y\right)^j g(t^{\frac{1}{m}}y)$,
$j=1,2,\ldots,m+n-1$ will vanish as $t\to 0+$.

We next consider the term involving $ g(t^{\frac{1}{m}}y)$, which is given by
$$ \sum_{j=0}^{m+n-1}(-1)^{j}\frac{b_j}{m^j}
\left(\Delta_x\right)^j f(t^{\frac{1}{m}}x)g(t^{\frac{1}{m}}y).$$
In view of the differential equation \eqref{eq:gDf} satisfied by $f(t^{\frac{1}{m}}x)$, it is readily seen that
\begin{multline}\label{eq:zerobehaf}
-\frac{\Delta_x}{m}\left[\sum_{j=0}^{m+n-1}(-1)^{j}\frac{b_j}{m^j}
\left(\Delta_x\right)^j f(t^{\frac{1}{m}}x)\right]=\prod_{j=0}^{m+n-1}\left(\frac{\Delta_x}{m}+\nu_j\right)f(t^\frac{1}{m}x)\\
=(-1)^m \frac{tx^m}{m^mn^n} f(t^{\frac{1}{m}}x).
\end{multline}
Note that
$$f(t^{\frac{1}{m}}x)=\sum_{i=0}^\infty c_i t^{\frac{\alpha+1-m+i}{m}}x^{\alpha+1-m+i}$$
for some constants $c_i$ with $c_0=\Gamma(\frac{(\alpha+1)n}{m})/m$. This, together with \eqref{eq:zerobehaf}, implies that
$$\sum_{j=0}^{m+n-1}(-1)^{j}\frac{b_j}{m^j}
\left(\Delta_x\right)^j f(t^{\frac{1}{m}}x)=\mathcal{O}(t^{1+\frac{\alpha+1-m}{m}}) \to 0 $$
as $t\to 0+$, provided the summation (essentially $b_0 f(t^{\frac{1}{m}}x)$ ) does not contain any non-zero constant term, i.e., $\alpha\neq m-1$. If $\alpha=m-1$, however, we see from \eqref{def:nui} that $\nu_n=0$, which implies that
$$b_0=\prod_{i=1}^{m+n-1}(-\nu_i)=0,$$
on account of the definition of $b_0$ in \eqref{def:ai}.

To show \eqref{eq:integrablereprscale}, we observe from the third equality in \eqref{eq:kinMeij} that
\begin{align}\label{eq:Kscaling}
  K^{(\alpha,\frac{m}{n})}(mn^{\frac{n}{m}}x^{\frac{1}{m}},mn^{\frac{n}{m}}y^{\frac{1}{m}})
& =\frac{x^{1-\frac{1}{m}}}{n^{\frac{n}{m}}} \int_0^1 G^{m,0}_{0,m+n}\left({- \atop -\nu_{m+n-1}, \ldots, -\nu_1, \nu_0} \Big{|}tx\right)
\nonumber
\\& \qquad \qquad \times G^{n,0}_{0,m+n}\left({- \atop \nu_0, \nu_1, \ldots, \nu_{m+n-1} }\Big{|}
ty \right)\ud t.
\end{align}
The first equality can then be proved in a manner similar to that of \eqref{eq:integrablerepr}. The second equality follows from the relations between Meijer-G functions and Wright's generalized Bessel functions, as indicated in \eqref{def:f} and \eqref{def:g}.

This completes the proof of Theorem \ref{prop:integrable}.
\qed

\section{Proofs of Propositions \ref{prop:PDEforxy}--\ref{prop:Hamiltonian} and Theorem \ref{prop:oneinterval}}\label{Sec:proofofPDE}

\subsection{Preliminaries}
It is the aim of this section to fix the notations and to collect some basic facts about the operators
for later use.

Throughout Section \ref{Sec:proofofPDE}, we shall denote by
\begin{align}\label{def:tilfg}
\widetilde f (x) = G^{m,0}_{0,m+n}\left({- \atop -\nu_{m+n-1}, \ldots, -\nu_1, \nu_0} \Big{|}x\right), \quad
\widetilde g (y) = G^{n,0}_{0,m+n}\left({- \atop \nu_0, \nu_1, \ldots, \nu_{m+n-1} } \Big{|} y\right),
\end{align}
where $\nu_0=0$ and $\nu_i$, $i\geq 1$ is given in \eqref{def:nui}. Recall the function $\widetilde{K}^{(\alpha,\frac{m}{n})}(x,y)$ defined in \eqref{def:tildeker}, the following proposition is immediate.
\begin{prop}
We have
\begin{equation}\label{eq:tildeKint}
\widetilde{K}^{(\alpha,\frac{m}{n})}(x,y)=\int_0^1 \widetilde f(tx) \widetilde g(ty)\ud t.
\end{equation}
Furthermore, the functions $\widetilde f$ and $\widetilde g$ satisfy the differential equations
\begin{align}
\prod_{j=0}^{m+n-1}\left(\Delta_x+\nu_j\right)\widetilde f(x)&=(-1)^m x \widetilde f(x), \label{eq:eqftilde}
\\
\prod_{j=0}^{m+n-1}\left(\Delta_y-\nu_j\right)\widetilde g(y)&=(-1)^n y \widetilde g(y),\label{eq:eqgtilde}
\end{align}
respectively.
\end{prop}
\begin{proof}
The equality \eqref{eq:tildeKint} follows from the fact that
$$\widetilde{K}^{(\alpha,\frac{m}{n})}(x,y)=\frac{n^{\frac{n}{m}}}{x^{1-\frac{1}{m}}}K^{(\alpha,\frac{m}{n})}(mn^{\frac{n}{m}}x^{\frac{1}{m}},mn^{\frac{n}{m}}y^{\frac{1}{m}}) $$
and \eqref{eq:Kscaling}, while \eqref{eq:eqftilde} and \eqref{eq:eqgtilde} can be seen from \eqref{def:tilfg} and \eqref{eq:diff}.
\end{proof}

Next, we come to the operators; cf. \cite{Forrester10,TW93} for more information. Let $\mathcal{K}$ be an operator depending smoothly on a parameter $s$ such that its derivative with respect to $s$ exists, it is readily seen that
\begin{equation}\label{eq:ders}
\frac{\ud}{\ud s}(I-\mathcal{K})^{-1}=(I-\mathcal{K})^{-1}\frac{\ud \mathcal{K}}{\ud s}(I-\mathcal{K})^{-1}.
\end{equation}
For any operator $\mathcal{L}$, the following commutator identity holds:
\begin{equation}\label{eq:LI-Kinv}
[\mathcal{L},(I-\mathcal{K})^{-1}]=(I-\mathcal{K})^{-1}[\mathcal{L},\mathcal{K}](I-\mathcal{K})^{-1},
\end{equation}
where $[\mathcal{L},\mathcal{K}] = \mathcal{L}\mathcal{K}-\mathcal{K}\mathcal{L}$ stands for the standard commutator of two operators. As usual, we set
\begin{equation}\label{def:MD}
\begin{aligned}
&\textrm{$M$: the operator of multiplication by the independent variable}, \\
&\textrm{$D$: the operator of differentiation}.
\end{aligned}
\end{equation}
If necessary, a subscript will be put on an operator to indicate the variable on which it acts.
Suppose that the operator $\mathcal{L}$ has the kernel $L(x,y)$, then
\begin{equation}\label{eq:MDK}
[MD,\mathcal{L}]\doteq ((MD)_x+(MD)_y+I)L(x,y),
\end{equation}
where $\doteq$ is understood as ``has kernel''.

With the operator $\widetilde {\mathcal{K}}_{\alpha,\frac{m}{n}}$ introduced in Section \ref{sec:FredholmDet}, we will denote by $\rho_{\alpha,\frac{m}{n}}$ the kernel of the operator $(I-\widetilde {\mathcal{K}}_{\alpha,\frac{m}{n}})^{-1}$,
and by $R_{\alpha,\frac{m}{n}}$ the kernel of the resolvent operator
$$(I-\widetilde {\mathcal{K}}_{\alpha,\frac{m}{n}})^{-1}-I=\widetilde {\mathcal{K}}_{\alpha,\frac{m}{n}}(I-\widetilde {\mathcal{K}}_{\alpha,\frac{m}{n}})^{-1}=(I-\widetilde {\mathcal{K}}_{\alpha,\frac{m}{n}})^{-1}\widetilde {\mathcal{K}}_{\alpha,\frac{m}{n}}.$$
By definitions, it is readily seen that
\begin{equation}\label{eq:Randrho}
\rho_{\alpha,\frac{m}{n}}(x,y)=R_{\alpha,\frac{m}{n}}(x,y)+\delta(x-y),
\end{equation}
where $\delta(x-y)$ is the Dirac delta function at $x = y$. The kernels $\rho_{\alpha,\frac{m}{n}}'$ and $R_{\alpha,\frac{m}{n}}'$ are
then defined similarly, which correspond to the operator $\widetilde {\mathcal{K}}_{\alpha,\frac{m}{n}}'$.
To this end, we also define $\widetilde {\mathcal{K}}_{\alpha,\frac{m}{n}}^t$ to be the transpose of the operator
$\widetilde {\mathcal{K}}_{\alpha,\frac{m}{n}}$, which acts on distributions and whose kernel is given
by $\widetilde {K}^{(\alpha,\frac{m}{n})}(y,x)\chi_{\widetilde J}(x)$. By definition, it is readily seen that
\begin{equation}\label{eq:transandprime}
(I-\widetilde {\mathcal{K}}^t_{\alpha,\frac{m}{n}})^{-1}(f\chi_{\widetilde J})(x)=(I-\widetilde {\mathcal{K}}'_{\alpha,\frac{m}{n}})^{-1}f(x)\chi_{\widetilde J}(x).
\end{equation}

Finally, we recall that since both $\widetilde {\mathcal{K}}_{\alpha,\frac{m}{n}}$ and $\widetilde {\mathcal{K}}_{\alpha,\frac{m}{n}}'$ are integrable in the sense of Its-Izergin-Korepin-Slavnov \cite{IIKS90}, their resolvent operators are integrable as well. Indeed, for $0 \leq j \leq m+n-1$, by introducing functions
\begin{align}
\mathcal{Q}_j(x;a)&=(I-\widetilde {\mathcal{K}}_{\alpha,\frac{m}{n}})^{-1}\phi_j(x), \label{def:Qxa}
\\
\mathcal{P}_j(y;a)&=(I-\widetilde {\mathcal{K}}'_{\alpha,\frac{m}{n}})^{-1}\psi_j(y), \label{def:Pxa}
\end{align}
where $a$ stands for the collection of parameters $a_1,a_2,\ldots,a_{2\ell}$, the functions $\phi_j$ and $\psi_j$ are given in \eqref{def:phi} and \eqref{def:psi}, respectively, we have the following proposition.
\begin{prop}\label{prop:resolexp}
Let $R_{\alpha,\frac{m}{n}}(x,y)$ and $R'_{\alpha,\frac{m}{n}}(x,y)$ be resolvent kernels of the operators $\widetilde {\mathcal{K}}_{\alpha,\frac{m}{n}}$ and $\widetilde {\mathcal{K}}'_{\alpha,\frac{m}{n}}$, respectively. Then, their explicit expressions are given below:
\begin{equation}\label{eq:explRmn}
R_{\alpha,\frac{m}{n}}(x,y)=\sum_{j=0}^{m+n-1}\frac{\mathcal{Q}_j(x;a)\mathcal{P}_j(y;a)}{x-y}\chi_{\widetilde J}(y),
\end{equation}
and
\begin{equation}\label{eq:explRprimemn}
R'_{\alpha,\frac{m}{n}}(x,y)=\sum_{j=0}^{m+n-1}\frac{\mathcal{Q}_j(y;a)\mathcal{P}_j(x;a)}{y-x}\chi_{\widetilde J}(y),
\end{equation}
where the functions $\mathcal{Q}_j$ and $\mathcal{P}_j$ are given in \eqref{def:Qxa} and \eqref{def:Pxa}, respectively.
\end{prop}
\begin{proof}
See \cite{IIKS90} or \cite[lemma 2.8]{DIZ}.
\end{proof}

In view of \eqref{def:xy}, \eqref{def:Qxa} and \eqref{def:Pxa}, we have that
\begin{equation}\label{eq:xyQP}
x_{j,k}=\lim_{x\in \widetilde{J},~ x \to a_k }\mathcal{Q}_j(x;a),\quad
y_{j,k}=\lim_{y\in \widetilde{J},~ y \to a_k }\mathcal{P}_j(y;a).
\end{equation}
Thus, it is natural to derive differential equations satisfied by $\mathcal{Q}_j$ and $\mathcal{P}_j$ in order to prove our theorems, which we will deal with in the next section.

\subsection{Partial differential equations for $\mathcal{Q}_j$ and $\mathcal{P}_j$ }

\begin{prop}\label{prop:pdeforQ}
With $\mathcal{P}_j(x;a)$ defined in \eqref{def:Pxa}, we set
\begin{equation}\label{def:wij}
w_{i,j}(a)=w_{i,j}(a_1,\ldots,a_{2\ell})=\int_{\widetilde{J}}\phi_i(x)\mathcal{P}_j(x;a) \ud x, \quad 0\leq i,j \leq m+n-1.
\end{equation}
The functions $\mathcal{Q}_j(x;a)$ defined in \eqref{def:Qxa} then satisfy the following system of partial differential equations: for $0\leq j<m+n-1$, we have
\begin{multline}\label{eq:Qpde1}
x\frac{\partial}{\partial x}\mathcal{Q}_j(x;a)
=(-1)^{n+1}w_{j,m+n-1}\mathcal{Q}_0(x;a)-\mathcal{Q}_{j+1}(x;a)
\\-\sum_{k=1}^{2\ell}(-1)^k a_k R_{\alpha,\frac{m}{n}}(x,a_k)\mathcal{Q}_j(a_k;a);
\end{multline}
for $j=m+n-1$, we have
\begin{multline}\label{eq:Qpde2}
x\frac{\partial}{\partial x}\mathcal{Q}_{m+n-1}(x;a)
=(-1)^{n+1}(x-w_{0,0}+w_{m+n-1,m+n-1})\mathcal{Q}_0(x;a)\\
+\sum_{i=1}^{m+n-1}\left((-1)^nw_{0,i}+b_{i-1}\right)\mathcal{Q}_i(x;a)
-\sum_{k=1}^{2\ell}(-1)^ka_kR_{\alpha,\frac{m}{n}}(x,a_k)\mathcal{Q}_{m+n-1}(a_k;a),
\end{multline}
where $b_i$ is given in \eqref{eq:elementary}.

Finally, for the derivative with respect to $a_k$, we have
\begin{equation}\label{eq:Qpde3}
\frac{\partial}{\partial a_k}\mathcal{Q}_j(x;a)=(-1)^kR_{\alpha,\frac{m}{n}}(x,a_k)\mathcal{Q}_j(a_k;a),\quad 0\leq j \leq m+n-1, \quad 1\leq k \leq 2 \ell.
\end{equation}
\end{prop}

\begin{proof}
With the operators $M$ and $D$ defined in \eqref{def:MD}, it is easily seen that
\begin{multline}\label{eq:MDQsplit}
x\frac{\partial}{\partial x}\mathcal{Q}_j(x;a)=MD\mathcal{Q}_j(x;a)=MD(I-\widetilde {\mathcal{K}}_{\alpha,\frac{m}{n}})^{-1}\phi_j(x)
\\
=[MD, (I-\widetilde {\mathcal{K}}_{\alpha,\frac{m}{n}})^{-1}]\phi_j(x)+(I-\widetilde {\mathcal{K}}_{\alpha,\frac{m}{n}})^{-1}MD\phi_j(x),
\end{multline}
for $j=0,1,\ldots,m+n-1$. Our strategy is then to evaluate the two parts in the last equality separately.

For the first one, by \eqref{eq:LI-Kinv}, it follows that
\begin{equation}\label{eq:MDI-Kinv}
[MD, (I-\widetilde {\mathcal{K}}_{\alpha,\frac{m}{n}})^{-1}]\phi_j(x)=(I-\widetilde {\mathcal{K}}_{\alpha,\frac{m}{n}})^{-1}[MD,\widetilde {\mathcal{K}}_{\alpha,\frac{m}{n}}](I-\widetilde {\mathcal{K}}_{\alpha,\frac{m}{n}})^{-1}\phi_j(x).
\end{equation}
In view of \eqref{eq:MDK}, we further have
\begin{align}
[MD,\widetilde {\mathcal{K}}_{\alpha,\frac{m}{n}}] & \doteq ((MD)_x+(MD)_y+I)\left(\widetilde{K}^{(\alpha,\frac{m}{n})}(x,y)\chi_{\widetilde J}(y)\right) \nonumber
\\
& = (\Delta_x+\Delta_y)\widetilde{K}^{(\alpha,\frac{m}{n})}(x,y)\cdot \chi_{\widetilde J}(y)+y\widetilde{K}^{(\alpha,\frac{m}{n})}(x,y)\sum_{k=1}^{2\ell}(-1)^{k-1}\delta(y-a_k)
\nonumber \\
&\qquad \qquad \qquad \qquad \qquad \qquad \qquad \qquad \qquad \qquad \quad \quad
+\widetilde{K}^{(\alpha,\frac{m}{n})}(x,y)\chi_{\widetilde J}(y),
\end{align}
where we have made use of the fact that
$$\frac{\partial}{\partial y}\chi_{\widetilde J}(y)=\sum_{k=1}^{2\ell}(-1)^{k-1}\delta(y-a_k).$$
A key observation here, due to \eqref{eq:tildeKint} and the integration by parts, is that
\begin{multline*}
(\Delta_x+\Delta_y)\widetilde{K}^{(\alpha,\frac{m}{n})}(x,y)=\int_0^1 t \frac{\ud}{\ud t}\left(\widetilde f(tx) \widetilde g(t y)\right)\ud t \\
= \widetilde f(x) \widetilde g(y)-\widetilde{K}^{(\alpha,\frac{m}{n})}(x,y)=(-1)^{n+1}\phi_0(x)\psi_{m+n-1}(y)-\widetilde{K}^{(\alpha,\frac{m}{n})}(x,y);
\end{multline*}
recall that the functions $\phi_i$ and $\psi_i$ are defined in \eqref{def:phi} and \eqref{def:psi}, respectively.
Hence,
\begin{equation}
[MD,\widetilde {\mathcal{K}}_{\alpha,\frac{m}{n}}]  \doteq
(-1)^{n+1}\phi_0(x)\psi_{m+n-1}(y) \chi_{\widetilde J}(y)-\sum_{k=1}^{2\ell}(-1)^{k}a_k\widetilde{K}^{(\alpha,\frac{m}{n})}(x,a_k)\delta(y-a_k).
\end{equation}
This, together with \eqref{eq:MDI-Kinv}, implies that
\begin{align}\label{eq:kerMDcommuInv}
&[MD, (I-\widetilde {\mathcal{K}}_{\alpha,\frac{m}{n}})^{-1}]
\nonumber
\\
&\doteq (-1)^{n+1}\mathcal{Q}_0(x;a)(I-\widetilde {\mathcal{K}}^t_{\alpha,\frac{m}{n}})^{-1}(\psi_{m+n-1}\chi_{\widetilde J})(y)
-\sum_{k=1}^{2\ell}(-1)^{k}a_kR_{\alpha,\frac{m}{n}}(x,a_k)\rho_{\alpha,\frac{m}{n}}(a_k,y)
\nonumber
\\
&=(-1)^{n+1}\mathcal{Q}_0(x;a)\mathcal{P}_{m+n-1}(y;a)\chi_{\widetilde J}(y)
-\sum_{k=1}^{2\ell}(-1)^{k}a_kR_{\alpha,\frac{m}{n}}(x,a_k)\rho_{\alpha,\frac{m}{n}}(a_k,y),
\end{align}
where we have made use of the fact that
$$ (I-\widetilde {\mathcal{K}}^t_{\alpha,\frac{m}{n}})^{-1}(\psi_{m+n-1}\chi_{\widetilde J})(y)=(I-\widetilde {\mathcal{K}}'_{\alpha,\frac{m}{n}})^{-1}\psi_{m+n-1}(y) \cdot \chi_{\widetilde J}(y)=\mathcal{P}_{m+n-1}(y;a)\chi_{\widetilde J}(y);$$
see \eqref{eq:transandprime} for the first equality. It is then immediate that
\begin{multline}\label{eq:splitI}
[MD, (I-\widetilde {\mathcal{K}}_{\alpha,\frac{m}{n}})^{-1}]\phi_j(x)
=(-1)^{n+1}w_{j,m+n-1}\mathcal{Q}_0(x;a)
\\-\sum_{k=1}^{2\ell}(-1)^{k}a_kR_{\alpha,\frac{m}{n}}(x,a_k)\mathcal{Q}_j(a_k;a), ~~ 0\leq j \leq m+n-1.
\end{multline}

For the function $(I-\widetilde {\mathcal{K}}_{\alpha,\frac{m}{n}})^{-1}MD\phi_j(x)$, we observe from \eqref{def:phi} and \eqref{def:tilfg} that, if $0\leq j < m+n-1$,
\begin{multline}\label{eq:secpartcase1}
(I-\widetilde {\mathcal{K}}_{\alpha,\frac{m}{n}})^{-1}MD\phi_j(x)=(I-\widetilde {\mathcal{K}}_{\alpha,\frac{m}{n}})^{-1}\Delta_x\left[(-1)^{n+1-j}(\Delta_x)^j \widetilde f \right](x)
\\
=(I-\widetilde {\mathcal{K}}_{\alpha,\frac{m}{n}})^{-1}\left[(-1)^{n+1-j}(\Delta_x)^{j+1} \widetilde f \right](x)
=-(I-\widetilde {\mathcal{K}}_{\alpha,\frac{m}{n}})^{-1} \phi_{j+1} (x)=-\mathcal{Q}_{j+1}(x;a).
\end{multline}
If $j=m+n-1$, we first obtain from \eqref{eq:eqftilde} and \eqref{def:ai} that
\begin{align*}
(-1)^m x \widetilde f(x)&=\prod_{j=0}^{m+n-1}\left(\Delta_x+\nu_j\right)\widetilde f(x)
=\left[ \sum_{i=0}^{m+n-1}(-1)^{m+n-1-i}b_i(\Delta_x)^{i+1}\right]\widetilde f(x)  \\
& = \left[ \sum_{i=0}^{m+n-2}(-1)^{m+n-1-i}b_i(\Delta_x)^{i+1}\right]\widetilde f(x)
+ (\Delta_x)^{m+n}\widetilde f(x),
\end{align*}
or equivalently, on account of \eqref{def:phi},
\begin{equation}
\Delta_x\phi_{m+n-1}(x)=(-1)^{n+1}x\phi_0(x)+\sum_{i=0}^{m+n-2}b_i\phi_{i+1}(x).
\end{equation}
Hence,
\begin{multline}\label{eq:caseIIphi}
(I-\widetilde {\mathcal{K}}_{\alpha,\frac{m}{n}})^{-1}MD\phi_{m+n-1}(x)=(I-\widetilde {\mathcal{K}}_{\alpha,\frac{m}{n}})^{-1}\left[\Delta_x \phi_{m+n-1} \right](x)
\\=(-1)^{n+1}(I-\widetilde {\mathcal{K}}_{\alpha,\frac{m}{n}})^{-1}M\phi_0(x)+\sum_{i=0}^{m+n-2}b_i\mathcal{Q}_{i+1}(x;a).
\end{multline}
It then remains to evaluate $(I-\widetilde {\mathcal{K}}_{\alpha,\frac{m}{n}})^{-1}M\phi_0(x)$ and a straightforward calculation yields
\begin{multline}\label{eq:I-KM}
(I-\widetilde {\mathcal{K}}_{\alpha,\frac{m}{n}})^{-1}M\phi_0(x)=[(I-\widetilde {\mathcal{K}}_{\alpha,\frac{m}{n}})^{-1}, M]\phi_0(x)+M(I-\widetilde {\mathcal{K}}_{\alpha,\frac{m}{n}})^{-1}\phi_0(x)
\\=[(I-\widetilde {\mathcal{K}}_{\alpha,\frac{m}{n}})^{-1}, M]\phi_0(x)+x\mathcal{Q}_0(x;a).
\end{multline}
By \eqref{eq:LI-Kinv}, one has
\begin{equation*}
[(I-\widetilde {\mathcal{K}}_{\alpha,\frac{m}{n}})^{-1}, M]=(I-\widetilde {\mathcal{K}}_{\alpha,\frac{m}{n}})^{-1}[\widetilde{\mathcal{K}}_{\alpha,\frac{m}{n}},M](I-\widetilde {\mathcal{K}}_{\alpha,\frac{m}{n}})^{-1},
\end{equation*}
and by \eqref{def:tildeker},
\begin{align*}
[\widetilde{\mathcal{K}}_{\alpha,\frac{m}{n}},M]\doteq \widetilde{K}^{(\alpha,\frac{m}{n})}(x,y)\chi_{\widetilde J}(y)\cdot y-x\cdot \widetilde{K}^{(\alpha,\frac{m}{n})}(x,y)\chi_{\widetilde J}(y)
=-\sum_{i=0}^{m+n-1}\phi_i(x)\psi_i(y)\chi_{\widetilde J}(y).
\end{align*}
Combining the above two formulas, we obtain that
\begin{multline}\label{eq:kerI-KinvM}
[(I-\widetilde {\mathcal{K}}_{\alpha,\frac{m}{n}})^{-1}, M] \doteq
-\sum_{i=0}^{m+n-1}\mathcal{Q}_i(x;a)(I-\widetilde {\mathcal{K}}^t_{\alpha,\frac{m}{n}})^{-1}(\psi_{i}\chi_{\widetilde J})(y)
\\=-\sum_{i=0}^{m+n-1}\mathcal{Q}_i(x;a)\mathcal{P}_i(y;a)\chi_{\widetilde J}(y).
\end{multline}
Inserting \eqref{eq:kerI-KinvM} into \eqref{eq:I-KM}, it is readily seen that
\begin{equation}\label{eq:kerI-KinvM2}
(I-\widetilde {\mathcal{K}}_{\alpha,\frac{m}{n}})^{-1}M\phi_0(x)=-\sum_{i=0}^{m+n-1}w_{0,i}\mathcal{Q}_i(x;a)+x\mathcal{Q}_0(x;a),
\end{equation}
which, by \eqref{eq:caseIIphi}, also gives us
\begin{multline}\label{eq:secpartcase2}
(I-\widetilde {\mathcal{K}}_{\alpha,\frac{m}{n}})^{-1}MD\phi_{m+n-1}(x)
=(-1)^{n+1}x\mathcal{Q}_0(x;a)+\sum_{i=1}^{m+n-1}b_{i-1}\mathcal{Q}_{i}(x;a)
\\
+(-1)^{n}\sum_{i=0}^{m+n-1}w_{0,i}\mathcal{Q}_i(x;a).
\end{multline}
The equation \eqref{eq:Qpde1} then follows from \eqref{eq:MDQsplit}, \eqref{eq:splitI} and \eqref{eq:secpartcase1}, while \eqref{eq:Qpde2} follows from \eqref{eq:MDQsplit}, \eqref{eq:splitI} and \eqref{eq:secpartcase2}.

To show \eqref{eq:Qpde3}, we note from \eqref{eq:ders} that
\begin{multline}\label{eq:pfQpde3}
\frac{\partial}{\partial a_k}\mathcal{Q}_j(x;a)=\frac{\partial}{\partial a_k}(I-\widetilde {\mathcal{K}}_{\alpha,\frac{m}{n}})^{-1}\phi_j(x)
=(I-\widetilde {\mathcal{K}}_{\alpha,\frac{m}{n}})^{-1}\frac{\partial}{\partial a_k}\widetilde {\mathcal{K}}_{\alpha,\frac{m}{n}}(I-\widetilde {\mathcal{K}}_{\alpha,\frac{m}{n}})^{-1}\phi_j(x),
\end{multline}
for any $0\leq j \leq m+n-1$ and $1\leq k \leq 2 \ell$. Since
\begin{align}\label{eq:kerKpartialak}
\frac{\partial}{\partial a_k}\widetilde {\mathcal{K}}_{\alpha,\frac{m}{n}}\doteq
\frac{\partial}{\partial a_k}\left(\widetilde{K}^{(\alpha,\frac{m}{n})}(x,y)\chi_{\widetilde J}(y)\right)=(-1)^k\widetilde{K}^{(\alpha,\frac{m}{n})}(x,y)\delta(y-a_k),
\end{align}
it then follows that
\begin{align}\label{eq:kerinvderiv}
(I-\widetilde {\mathcal{K}}_{\alpha,\frac{m}{n}})^{-1}\frac{\partial}{\partial a_k}\widetilde {\mathcal{K}}_{\alpha,\frac{m}{n}}(I-\widetilde {\mathcal{K}}_{\alpha,\frac{m}{n}})^{-1} \doteq (-1)^k R_{\alpha,\frac{m}{n}}(x,a_k)\rho_{\alpha,\frac{m}{n}}(a_k,y).
\end{align}
Inserting the above formula into \eqref{eq:pfQpde3} gives us \eqref{eq:Qpde3}.

This completes the proof of Proposition \ref{prop:pdeforQ}.
\end{proof}

The system of partial differential equations satisfied by the functions $\mathcal{P}_j$ are given in the next proposition.
\begin{prop}\label{prop:pdeforP}
For $j=0$, we have
\begin{multline}\label{eq:Ppde2}
y\frac{\partial}{\partial y}\mathcal{P}_{0}(y;a)
=(-1)^n\sum_{i=0}^{m+n-2}w_{i,m+n-1}\mathcal{P}_i(y;a)+(-1)^n(y-w_{0,0}+w_{m+n-1,m+n-1})\mathcal{P}_{m+n-1}(y;a)
\\
-\sum_{k=1}^{2\ell}(-1)^ka_kR'_{\alpha,\frac{m}{n}}(y,a_k)\mathcal{P}_{0}(a_k;a);
\end{multline}
for $1\leq j \leq m+n-1$, we have
\begin{multline}\label{eq:Ppde1}
y\frac{\partial}{\partial y}\mathcal{P}_j(y;a)
=\mathcal{P}_{j-1}(y;a)+((-1)^{n+1}w_{0,j}-b_{j-1})\mathcal{P}_{m+n-1}(y;a)
\\
-\sum_{k=1}^{2\ell}(-1)^k a_k R'_{\alpha,\frac{m}{n}}(y,a_k)\mathcal{P}_j(a_k;a).
\end{multline}

Finally, for the derivative with respect to $a_k$, we have
\begin{equation}\label{eq:Ppde3}
\frac{\partial}{\partial a_k}\mathcal{P}_j(y;a)=(-1)^kR'_{\alpha,\frac{m}{n}}(y,a_k)\mathcal{P}_j(a_k;a),\quad 0\leq j \leq m+n-1, \quad 1\leq k \leq 2 \ell.
\end{equation}
\end{prop}
\begin{proof}
The proof is similar to that of Proposition \ref{prop:pdeforQ}. The main difference is that we need to use the following fact:
\begin{align}
&MD\psi_j(y)
\nonumber
\\
&=\Delta_y\left[\sum_{i=0}^{m+n-1-j}b_{i+j}(\Delta_y)^i \widetilde g \right](y)=
\sum_{i=0}^{m+n-1-j}b_{i+j}(\Delta_y)^{i+1} \widetilde g(y)
\nonumber
\\
&=\left\{
    \begin{array}{ll}
     \sum_{i=1}^{m+n-j}b_{i+j-1}(\Delta_y)^{i} \widetilde g(y)=\psi_{j-1}(y)-b_{j-1}\psi_{m+n-1}(y), & \hbox{$1\leq j \leq m+n-1$,} \\
     \prod_{j=0}^{m+n-1}\left(\Delta_y-\nu_j\right)\widetilde g(y)=(-1)^ny\psi_{m+n-1}(y), & \hbox{$j=0$,}
    \end{array}
  \right.
\end{align}
which can be seen from \eqref{def:psi}, \eqref{def:tilfg}, \eqref{def:ai} and \eqref{eq:eqgtilde}. This also explains why the results are stated for $j=0$ and $1\leq j \leq m+n-1$, respectively. We leave the details to interested readers.

This completes the proof of Proposition \ref{prop:pdeforP}.
\end{proof}

Finally, we also need to relate the partial derivatives of $w_{i,j}$ and $R_{\alpha,\frac{m}{n}}$ to $\mathcal{P}_j$ and $\mathcal{Q}_j$. The relations are summarized in the following proposition.
\begin{prop}\label{prop:pdeforR}
With $w_{i,j}$ defined in \eqref{def:wij}, we have
\begin{equation}\label{eq:partialw}
\frac{\partial}{\partial a_k}w_{i,j}=(-1)^k\mathcal{Q}_{i}(a_k;a)\mathcal{P}_{j}(a_k;a), ~~1\leq k \leq 2\ell, ~~0\leq i,j \leq m+n-1.
\end{equation}
Furthermore, the resolvent kernel $R_{\alpha,\frac{m}{n}}(x,y)$ of the operator $\widetilde {\mathcal{K}}_{\alpha,\frac{m}{n}}$ satisfies
\begin{equation}\label{eq:partialR1}
\frac{\partial}{\partial a_k}R_{\alpha,\frac{m}{n}}(x,y)=(-1)^k R_{\alpha,\frac{m}{n}}(x,a_k)\rho_{\alpha,\frac{m}{n}}(a_k,y), \qquad 1 \leq k \leq 2\ell
\end{equation}
and
\begin{equation}\label{eq:partialR2}
\left(\Delta_x+\Delta_y+\sum_{k=1}^{2\ell}\Delta_{a_k}+I \right) R_{\alpha,\frac{m}{n}}(x,y)
=(-1)^{n+1}\mathcal{Q}_0(x)\mathcal{P}_{m+n-1}(y)\chi_{\widetilde J}(y).
\end{equation}
\end{prop}
\begin{proof}
By \eqref{def:wij}, we obtain from \eqref{eq:Ppde3}, Proposition \ref{prop:resolexp}, \eqref{def:Qxa} and a straightforward calculation
that
\begin{align}
\frac{\partial}{\partial a_k}w_{i,j}&=(-1)^k\mathcal{P}_j(a_k;a)\left(\phi_{i}(a_k)+\int_{\widetilde{J}}\phi_i(x)R'_{\alpha,\frac{m}{n}}(x,a_k)\ud x\right)
\nonumber \\
&= (-1)^k\mathcal{P}_j(a_k;a)\cdot \int_0^\infty \phi_i(x)(R_{\alpha,\frac{m}{n}}(a_k,x)+\delta(x-a_k))\ud x
= (-1)^k\mathcal{Q}_{i}(a_k;a)\mathcal{P}_{j}(a_k;a), \nonumber
\end{align}
where in the last step we have used \eqref{eq:Randrho} which asserts that
\begin{equation*}
(I-\widetilde {\mathcal{K}}_{\alpha,\frac{m}{n}})^{-1} \doteq R_{\alpha,\frac{m}{n}}(x,y)+\delta(x-y).
\end{equation*}
Thus,
$\frac{\partial}{\partial a_k}(I-\widetilde {\mathcal{K}}_{\alpha,\frac{m}{n}})^{-1}\doteq \frac{\partial}{\partial a_k}R_{\alpha,\frac{m}{n}}(x,y). $
This, together with \eqref{eq:ders} and \eqref{eq:kerinvderiv}, gives us \eqref{eq:partialR1}.

Finally, we see from \eqref{eq:MDK} that
\begin{equation}
[MD, (I-\widetilde {\mathcal{K}}_{\alpha,\frac{m}{n}})^{-1}\widetilde {\mathcal{K}}_{\alpha,\frac{m}{n}}]\doteq
\left(\Delta_x+\Delta_y+I \right) R_{\alpha,\frac{m}{n}}(x,y).
\end{equation}
On the other hand, it is easily seen that
\begin{equation}
[MD, (I-\widetilde {\mathcal{K}}_{\alpha,\frac{m}{n}})^{-1}\widetilde {\mathcal{K}}_{\alpha,\frac{m}{n}}]
=[MD, (I-\widetilde {\mathcal{K}}_{\alpha,\frac{m}{n}})^{-1}-I]=[MD, (I-\widetilde {\mathcal{K}}_{\alpha,\frac{m}{n}})^{-1}].
\end{equation}
A combination of the above two formulas and \eqref{eq:kerMDcommuInv} shows that
\begin{multline}
\left(\Delta_x+\Delta_y+I \right) R_{\alpha,\frac{m}{n}}(x,y)
=(-1)^{n+1}\mathcal{Q}_0(x;a)\mathcal{P}_{m+n-1}(y;a)\chi_{\widetilde J}(y)\\-\sum_{k=1}^{2\ell}(-1)^{k}a_kR_{\alpha,\frac{m}{n}}(x,a_k)\rho_{\alpha,\frac{m}{n}}(a_k,y),
\end{multline}
which, in view of \eqref{eq:partialR1}, is equivalent to \eqref{eq:partialR2}.

This completes the proof of Proposition \ref{prop:pdeforR}.
\end{proof}

\subsection{Proofs of Propositions \ref{prop:PDEforxy} and \ref{prop:Hamiltonian}}
\paragraph{Proof of Proposition \ref{prop:PDEforxy}}
In view of \eqref{eq:xyQP}, it is readily seen from \eqref{eq:Qpde3} and \eqref{eq:explRmn} that
\begin{equation*}
\frac{\partial x_{j,k}}{\partial a_i}=(-1)^i R_{\alpha,\frac{m}{n}}(a_k,a_i)x_{j,i}
=(-1)^i \frac{x_{j,i}}{a_k-a_i}\sum_{l=0}^{m+n-1}x_{l,k}y_{l,i},
\end{equation*}
for $1\leq k \neq i \leq 2 \ell$ and $0\leq j \leq m+n-1$, which is \eqref{eq:xjkpi}. If $k=i$, we have
\begin{equation}\label{case:k=i}
\frac{\partial x_{j,k}}{\partial a_k}=\left(\frac{\partial}{\partial x}+\frac{\partial}{\partial a_k}\right)\mathcal{Q}_j(x;a)\Big{|}_{x=a_k}.
\end{equation}
Since
\begin{equation}\label{eq:vinw}
v_j=(-1)^nw_{j,m+n-1}, \qquad 0 \leq j \leq m+n-1,
\end{equation}
it then follows from \eqref{case:k=i}, \eqref{eq:Qpde1} and \eqref{eq:Qpde3} that, for $0\leq j \leq m+n-2$,
\begin{align*}
a_k\frac{\partial x_{j,k}}{\partial a_k}&=-v_jx_{0,k}-x_{j+1,k}-\sum_{i=1,i\neq k}^{2\ell}(-1)^ia_iR_{\alpha,\frac{m}{n}}(a_k,a_i)x_{j,i}
\\
&=-v_jx_{0,k}-x_{j+1,k}-\sum_{i=1, i\neq k}^{2\ell}(-1)^i\frac{a_i x_{j,i}}{a_k-a_i}\sum_{l=0}^{m+n-1}x_{l,k}y_{l,i},
\end{align*}
which is \eqref{eq:xjkpk1}. The equation \eqref{eq:xjkpk2} can be proved similarly with the aid of \eqref{eq:Qpde2} and the fact that
\begin{equation}\label{eq:uinw}
u_j=(-1)^nw_{0,j}+b_{j-1}, \qquad 0 \leq j \leq m+n-1.
\end{equation}
The equation \eqref{eq:ujpartial} for $u_j$ follows directly from \eqref{eq:uinw}, \eqref{eq:partialw} and \eqref{eq:xyQP}.

On account of Proposition \ref{prop:pdeforP}, \eqref{eq:vinw} and \eqref{eq:partialw}, the other equations for $y_{j,k}$ and $v_j$ can be proved in a manner similar, and we omit the details.

This completes the proof of Proposition \ref{prop:PDEforxy}.

\paragraph{Proof of Proposition \ref{prop:Hamiltonian}}
To show \eqref{eq:Hkexplicit}, we first note that
$$\frac{\partial}{\partial a_k} \log \det \left(I- \widetilde{\mathcal{K}}_{\alpha,\frac{m}{n}} \right )=- \tr \left((I-\widetilde{\mathcal{K}}_{\alpha,\frac{m}{n}})^{-1}\frac{\partial \widetilde{\mathcal{K}}_{\alpha,\frac{m}{n}} }{\partial a_k}\right), \quad 1\leq k \leq 2\ell,$$
and by \eqref{eq:kerKpartialak},
$$(I-\widetilde{\mathcal{K}}_{\alpha,\frac{m}{n}})^{-1}\frac{\partial \widetilde{\mathcal{K}}_{\alpha,\frac{m}{n}} }{\partial a_k}
\doteq (-1)^k R_{\alpha,\frac{m}{n}}(x,y)\delta(y-a_k). $$
Thus, it is readily seen that
\begin{equation}\label{eq:partiallogdet}
\frac{\partial}{\partial a_k} \log \det \left(I- \widetilde{\mathcal{K}}_{\alpha,\frac{m}{n}} \right )=(-1)^{k+1}R_{\alpha,\frac{m}{n}}(a_k,a_k).
\end{equation}
Since
$$R_{\alpha,\frac{m}{n}}(a_k,a_k)=\sum_{j=0}^{m+n-1}\mathcal{Q}_j'(a_k;a)\mathcal{P}_j(a_k,a),$$
it then follows from \eqref{eq:xyQP}, \eqref{eq:Qpde1}, \eqref{eq:Qpde2}, \eqref{eq:vinw} and \eqref{eq:uinw} that
\begin{multline}\label{eq:akRinxy}
a_kR_{\alpha,\frac{m}{n}}(a_k,a_k)=\left(-\sum_{j=0}^{m+n-1}v_jy_{j,k}+(-1)^{n+1}a_k y_{m+n-1,k}\right)x_{0,k}+\sum_{j=0}^{m+n-1}u_jx_{j,k}y_{m+n-1,k}
\\-\sum_{j=0}^{m+n-2}x_{j+1,k}y_{j,k}-\sum_{j=0}^{m+n-1}\sum_{i=1,i\neq k}^{2\ell}(-1)^i a_i R_{\alpha,\frac{m}{n}}(a_k,a_i) x_{j,i}y_{j,k}.
\end{multline}
By the change of variables \eqref{eq:xytopq}, it is also readily seen that
$$x_{j,k}y_{i,k}=(-1)^{k+1}q_{j,k}p_{i,k}.$$
This, together with \eqref{eq:partiallogdet} and \eqref{eq:akRinxy}, gives us \eqref{eq:Hkexplicit}.

The other equations \eqref{eq:involution}--\eqref{eq:uvpartialder} then follows from straightforward calculations, with the aid of Proposition \ref{prop:PDEforxy}, \eqref{eq:xytopq} and \eqref{eq:Hkexplicit}.

This completes the proof of Proposition \ref{prop:Hamiltonian}.

\subsection{Proof of Theorem \ref{prop:oneinterval}}
It is easily seen that the equations \eqref{eq:eqforx}--\eqref{eq:ICuv} correspond to Proposition \ref{prop:PDEforxy}
with $\widetilde {J}=(0,s)$, i.e., $\ell=1$, $a_1=0$ and $a_2=s$. To show \eqref{eq:Gapformular}, we first establish the relevant results for the Fredholm determinant $\det \left(I- \widetilde{\mathcal{K}}_{\alpha,\frac{m}{n}}\Big{|}_{(0,s)} \right )$.

By \eqref{eq:partiallogdet}, it follows that
\begin{equation*}
\frac{\ud}{\ud s} \log \det \left(I- \widetilde{\mathcal{K}}_{\alpha,\frac{m}{n}}\Big{|}_{(0,s)} \right) =-R(s),
\end{equation*}
where $R(s):=R_{\alpha,\frac{m}{n}}(s,s)$. Since we are dealing with the special case $\widetilde {J}=(0,s)$, it is also  immediate from \eqref{eq:partialR2} and \eqref{eq:xyQP} that
$$(sR(s))'=R(s)+s\frac{\ud}{\ud s}R(s)=(-1)^{n+1}x_0(s)y_{m+n-1}(s).$$
A combination of the above two formulas implies that
\begin{align}\label{eq:logdettilde}
&\log \det \left(I- \widetilde{\mathcal{K}}_{\alpha,\frac{m}{n}}\Big{|}_{(0,s)} \right)
=-\int_0^s R(\varrho)\ud \varrho=(-1)^n\int_0^s\frac{1}{\varrho}\int_{0}^\varrho x_0(t)y_{m+n-1}(t)\ud t \ud \varrho
\nonumber
\\
&=(-1)^n\int_0^s\int_t^s \frac{1}{\varrho}\ud \varrho \cdot x_0(t)y_{m+n-1}(t) \ud t=
(-1)^n\int_0^s \log\left(\frac{s}{t}\right) x_0(t)y_{m+n-1}(t) \ud t
\nonumber
\\
&=\int_0^s \frac{v_0(t)}{t} \ud t,
\end{align}
where we have made use of \eqref{eq:eqforv} and the integration by parts in the last step.
In view of \eqref{eq:relKtoKtilde}, we have that
\begin{equation}\label{eq:FKtoFktilde}
F_{\alpha,\frac{m}{n}}(s)=\det \left(I- \mathcal{K}_{\alpha,\frac{m}{n}}\Big{|}_{(0,s)} \right)=\det \left(I- \widetilde{\mathcal{K}}_{\alpha,\frac{m}{n}}\Big{|}_{(0,\frac{s^m}{m^mn^n})} \right).
\end{equation}
Combining the above two formulas and a change of variable $t \to \frac{\tau^m}{m^mn^n}$ gives us \eqref{eq:Gapformular}.

This completes the proof of Theorem \ref{prop:oneinterval}.

\appendix

\section{The Meijer G-function and Wright's generalized Bessel function}
\subsection{The Meijer G-function}
By definition, the Meijer G-function is given by the
following contour integral in the complex plane:
\begin{equation}\label{def:Meijer}
G^{m,n}_{p,q}\left({a_1,\ldots,a_p \atop b_1,\ldots,b_q}\Big{|}
z\right)
=\frac{1}{2\pi i}\int_\gamma
\frac{\prod_{j=1}^m\Gamma(b_j+u)\prod_{j=1}^n\Gamma(1-a_j-u)}
{\prod_{j=m+1}^q\Gamma(1-b_j-u)\prod_{j=n+1}^p\Gamma(a_j+u)}z^{-u}
\ud u,
\end{equation}
where $\Gamma$ denotes the usual gamma function and the branch cut
of $z^{-u}$ is taken along the negative real axis. It is also
assumed that
\begin{itemize}
  \item $0\leq m\leq q$ and $0\leq n \leq p$, where $m,n,p$ and $q$
  are integer numbers;
  \item The real or complex parameters $a_1,\ldots,a_p$ and
  $b_1,\ldots,b_q$ satisfy the conditions
  \begin{equation*}
  a_k-b_j \neq 1,2,3, \ldots, \quad \textrm{for $k=1,2,\ldots,n$ and $j=1,2,\ldots,m$,}
  \end{equation*}
  i.e., none of the poles of $\Gamma(b_j+u)$, $j=1,2,\ldots,m$ coincides
  with any poles of $\Gamma(1-a_k-u)$, $k=1,2,\ldots,n$.
\end{itemize}
The contour $\gamma$ is chosen in such a way that all the poles of
$\Gamma(b_j+u)$, $j=1,\ldots,m$ are on the left of the path, while
all the poles of $\Gamma(1-a_k-u)$, $k=1,\ldots,n$ are on the right,
which is usually taken to go from $-i\infty$ to $i\infty$. Most of the known special functions can be viewed as special cases of the Meijer G-functions. For more details, we refer to the references \cite{Luke,DLMF}.

From the definition \eqref{def:Meijer}, it is easily seen that
   \begin{equation}\label{eq:multiply}
   z^{\alpha}G^{m,n}_{p,q}\left({a_1,\ldots,a_p \atop b_1,\ldots,b_q}\Big{|}z\right)=G^{m,n}_{p,q}\left({a_1+\alpha,\ldots,a_p+\alpha \atop b_1+\alpha,\ldots,b_q+\alpha}\Big{|}z\right),
   \end{equation}

Finally, we note that the Meijer G-function satisfies the following linear
differential equation of order $\max(p,q)$:
\begin{equation}\label{eq:diff}
\bigg[(-1)^{p-m-n}z\prod_{j=1}^{p}\left(z\frac{\ud}{\ud
z}-a_j+1\right)
-\prod_{j=1}^{q}\left(z\frac{\ud}{\ud
z}-b_j\right)\bigg]G^{m,n}_{p,q}\left({a_1,\ldots,a_p \atop b_1,\ldots,b_q}\Big{|}
z\right)=0;
\end{equation}
see \cite[formula 16.21.1]{DLMF}.

\subsection{Wright's generalized Bessel function}
The Wright's generalized Bessel function $J_{a,b}$ defined in \eqref{eq:Wright} is an entire function of $z$ depending on two parameters $a$ and $b$. It generalizes the Bessel function due to the relation
$$J_{\nu+1,1}\left(\frac{z^2}{4}\right)=\left(\frac{z}{2}\right)^{-\nu}J_{\nu}(z),$$
where $J_{\nu}$ stands for the Bessel function of the first kind.
The Wright's generalized Bessel functions are related to the Meijer G-functions if $b$ is a rational number. Indeed, by \cite[formula (13)]{GLM00} and \cite[formula (22)]{GLM99}, it follows that
\begin{multline}\label{eq:WrightinMeijer}
J_{a,\frac{m}{n}}(z)=(2\pi)^{\frac{m-n}{2}}m^{-a+\frac{1}{2}}n^{\frac{1}{2}}\\G^{n,0}_{0,m+n}\left({- \atop 0,\frac{1}{n},\ldots,\frac{n-1}{n},1-\frac{a}{m},1-\frac{a}{m}-\frac{1}{m},\ldots,1-\frac{a}{m}-\frac{m-1}{m} }\Big{|}
\frac{z^n}{m^mn^n}\right),
\end{multline}
where $m,n\in\mathbb{N}$. This, together with \eqref{eq:diff}, implies that $J_{a,\frac{m}{n}}$ satisfies the following linear differential equation of order $m+n$:
\begin{equation}\label{eq:diffofJ}
\prod_{j=0}^{n-1}\left(\frac{z}{n}\frac{\ud}{\ud
z}-\frac{j}{n}\right)\prod_{i=0}^{m-1}\left(\frac{z}{n}\frac{\ud}{\ud
z}-1+\frac{a+i}{m}\right)J_{a,\frac{m}{n}}(z)=(-1)^n\frac{z^n}{m^mn^n}J_{a,\frac{m}{n}}(z).
\end{equation}

%

\section*{Acknowledgment}
The author thanks the anonymous referees for their careful reading and constructive suggestions. This work is partially supported by The Program for Professor of Special Appointment (Eastern Scholar) at Shanghai Institutions of Higher Learning (No. SHH1411007), by National Natural Science Foundation of China (No. 11501120) and by Grant EZH1411513 from Fudan University.



\end{document}